\documentclass[journal]{IEEEtran}

\usepackage{mathnotation} 
\usepackage{cite}
\usepackage{bbm}
\usepackage{makecell} 
\usepackage{graphicx}  
\usepackage[caption=false]{subfig}
\usepackage{tikz}
\usepackage{xspace}
\usepackage[hidelinks]{hyperref} 
\usepackage[capitalize,noabbrev]{cleveref}
\usepackage{booktabs}
\usetikzlibrary{positioning}
\usepackage{enumitem}

\def \alg{\texttt{DEFINED}\xspace}

\newcommand{\Prob}{\mathbb{P}}
\newcommand{\probp}[1]{\mathbb{P}\left(#1\right)}

\newcommand{\iid}{\stackrel{\mathrm{ i.i.d.}}{\sim}}
\newcommand{\pth}[1]{\left( #1 \right)}

\crefname{assumption}{Assumption}{Assumptions}

\newcommand{\test}{\mathsf{test}}

\newcommand{\query}{\mathsf{query}}
\newcommand{\out}{\mathsf{out}}

\def\E{\mathbb{E}}

\newcommand{\dis}{\mathcal{P}}
\def\testx{x}
\def\testy{y}

\newcommand{\normal}{\mathsf{N}}

\ifodd 1
\newcommand{\congr}[1]{{\color{blue}#1}}
\else
\newcommand{\congr}[1]{#}
\fi

\ifodd 1
\newcommand{\congc}[1]{{\color{red}(Cong: #1)}}
\else
\newcommand{\congc}[1]{}
\fi

\ifodd 1
\newcommand{\jing}[1]{{\color{red}#1}}
\else
\newcommand{\jing}[1]{#}
\fi

\ifodd 1
\newcommand{\jingc}[1]{{\color{red}(JY: #1)}}
\else
\newcommand{\jingc}[1]{}
\fi

\ifodd 1
\newcommand{\weir}[1]{{\color{orange}#1}}
\else
\newcommand{\weir}[1]{#}
\fi

\ifodd 1
\newcommand{\li}[1]{{\color[HTML]{007FFF}#1}} 
\else
\newcommand{\li}[1]{#}
\fi

\ifodd 1
\newcommand{\lic}[1]{{\color[HTML]{007FFF}(Li: #1)}} 
\else
\newcommand{\lic}[1]{}
\fi

\title{Decision Feedback In-Context Learning for Wireless Symbol Detection}

\author{Li~Fan,~\IEEEmembership{Graduate Student Member,~IEEE,}    
        Wei~Shen,~\IEEEmembership{Graduate Student Member,~IEEE,}\\
        Jing~Yang,~\IEEEmembership{Senior Member,~IEEE,}
        and~Cong~Shen~\IEEEmembership{Senior Member,~IEEE}
\thanks{A preliminary version of this work was presented at the 2025 IEEE International Conference on Communications \cite{fan2025icc}. Simulation code and examples are available at \texttt{\url{https://github.com/ShenGroup/DEFINED}}.}
\thanks{The authors are with the Charles L. Brown Department of Electrical and Computer Engineering, University of Virginia, USA. (E-mail: \texttt{\{lf2by, zyy5hb, yangjing, cong\}@virginia.edu}.)
}
}

\begin{document}

\maketitle

\begin{abstract}
Pre-trained Transformers, through in-context learning (ICL), have demonstrated exceptional capabilities to adapt to new tasks using example prompts \textit{without model update}. Transformer-based wireless receivers, where prompts consist of the pilot data in the form of transmitted and received signal pairs, have shown high detection accuracy when pilot data are abundant. However, pilot information is often costly and limited in practice. In this work, we propose \underline{DE}cision \underline{F}eedback \underline{IN}-Cont\underline{E}xt \underline{D}etection (\alg) as a new wireless receiver design, which bypasses channel estimation and directly performs symbol detection using the (sometimes extremely) limited pilot data. The key innovation in \alg is the proposed decision feedback mechanism in ICL, where we sequentially incorporate the detected symbols into the prompts as \emph{pseudo-labels} to improve the detection for subsequent symbols. We further establish an error lower bound and provide theoretical insights into the model's generalization under channel distribution mismatch. Extensive experiments across a broad range of wireless settings demonstrate that a small Transformer trained with \alg achieves significant performance improvements over conventional methods, in some cases only needing a single pilot pair to achieve similar performance to the latter with more than 4 pilot pairs. 

\end{abstract}


\section{Introduction}


Wireless receiver symbol detection aims to recover transmitted symbols over fading channels with varying SNRs. Traditional methods use a two-step process: channel estimation (e.g., via Minimum Mean Square Error (MMSE)) followed by symbol detection. However, this approach can be computationally intensive and sensitive to estimation errors, especially in high-dimensional systems. In contrast, machine learning-based approaches, including deep learning models, enable end-to-end learning that better adapts to complex, nonlinear environments and system imperfections, while supporting joint optimization across transceiver components. Such data-driven designs align with the vision of AI-native 6G systems \cite{hoydis2021toward}, where learning-based architectures are integrated from the ground up to enhance adaptability, efficiency, and performance.

Despite their potential, deep learning (DL) based solutions have not been widely deployed in real-world wireless systems. Among the many challenges, data dependency and generalization are among the most significant obstacles \cite{simeone2020learning}.  DL models require large amounts of diverse and high-quality training data to train, which is difficult to obtain in many wireless deployments. Poor generalization to new, unseen wireless environments (e.g., new channel conditions, interference patterns) is a particularly significant challenge due to the diverse use cases in wireless systems.  Although meta-learning \cite{finn2017model} can help with adapting to new tasks, existing approaches \cite{chen2023learning, park2020learning} often rely on explicit model parameter updates, which increase computational complexity and reduce robustness during adaptation. 



Advances in Transformer models \cite{vaswani2017attention}, particularly decoder-only architectures like GPT \cite{radford2019language}, have driven remarkable progress in natural language processing and a broad spectrum of other domains.  Recent work \cite{NEURIPS2020_1457c0d6} demonstrates that pre-trained Transformers can generalize to new tasks during inference through in-context learning (ICL), without requiring explicit model updates.  Specifically, the input is structured as 
\((y_1, f(y_1), \ldots, y_n, f(y_n), y_{\text{query}})\), where \((y_1, \ldots, y_n)\) represent features in the problem domain and \(f\) is an unknown mapping. By leveraging the Transformer's advanced next-token prediction capability, a well-pretrained model can approximate \(f(y_{\text{query}})\) with high accuracy for many classes of functions, conditioned on the given context. 

The wireless symbol detection problem involves estimating the transmitted symbol from the received signal, where the mapping between the transmitted and received signals can be modeled as a noisy function. This formulation fits naturally with the ICL framework and Transformer-based sequence models.
\cite{teja2023transformers} introduces Transformers for this task using ICL, framing it as a regression problem with MSE loss and achieving near-MMSE performance.  Subsequent work \cite{zecchin2024context} extends this approach to multiple-input multiple-output (MIMO) systems, emphasizing the importance of task diversity during pre-training. Similarly, \cite{zecchin2024cell} demonstrates that Transformers are robust to pilot contamination issues in multi-user MIMO systems. While these studies focus on training relatively small Transformers from scratch, \cite{abbas2024leveraging} proposes to leverage large language models (LLMs) to apply in-context learning. 
This method transforms the detection problem from a numerical task into a linguistic one, incorporating LLM calibration techniques. Collectively, these studies highlight that Transformers are powerful and versatile models, well-suited for addressing a range of challenges in wireless communication systems.

Despite these advances, prior studies face several limitations. Most solutions treat symbol detection as a regression task, relying on MSE-based objectives and post-processing to map continuous outputs to discrete symbols. Furthermore, many methods require an abundance of pilot pairs to achieve reasonable performance, which may not be feasible in practice. Additionally, the use of large models incurs long delay and requires significant memory and computation resources, further limiting their applicability in real-world deployments.

Inspired by decision feedback in wireless communication (e.g., decision feedback equalizers over multi-path fading channels), we propose a novel \textit{prompt design} by incorporating decision pairs. Specifically, we combine the current received signal with the model’s detected symbol in a pair, merging them with previous prompts to form a new, larger prompt for subsequent symbol detection. Our \alg model employs a carefully designed mixture training process to achieve high performance with limited pilots (sometimes as few as a single pilot) while maintaining accuracy when sufficient pilots are available. This {result} highlights the promising potential of Transformers in future wireless communication systems. Extensive experiments across various modulation schemes {and channel distributions} validate the effectiveness of our approach.
We summarize our contributions as follows:

\begin{itemize}[leftmargin=*]\itemsep=0pt
\item We propose a novel Transformer model that jointly performs channel estimation and symbol detection. Our key innovation in \alg is the incorporation of decision feedback to allow decision feedback to be used in conjunction with the limited pilot data, to gradually increase the effective context length and thus improve the output performance.

\item We develop a mixture training process for \alg, achieving significant performance improvements with very limited pilot data while maintaining high accuracy with sufficient pilot data, making the model adaptable to practical scenarios.

\item Theoretically, we derive an error lower bound and provide theoretical insights into the model's generalization capability under channel distribution mismatch.

\item We conduct extensive experiments across various modulation schemes under varying configurations, where we show that a very small Transformer model trained by \alg can achieve strong accuracy and generalization performance.

\end{itemize}

The remainder of this paper is organized as follows. Section~\ref{sec:related} reviews the related work. Section~\ref{sec:system model} presents the system model and canonical methods. Section~\ref{sec:ICL_detection} introduces in-context learning-based symbol detection. The proposed \alg method, including the model structure and training details, is described in Section~\ref{sec:DEFINED}. Section~\ref{sec:theory} provides the theoretical analysis of \alg. Section~\ref{sec:expriment} presents experimental results, and Section~\ref{sec:conc} concludes the paper.

\section{Related Work}
\label{sec:related}

Machine learning is playing a pivotal role in shaping wireless system design and is a key enabler of AI-native 6G architectures \cite{Liang2018jstsp,Liang2020tc,simeone2018very,yang2024dyspan,yang2024twc,yang2024asilomar}. 
The Transformer model \cite{vaswani2017attention} has become fundamental across various fields. GPT-2 extends Transformers for multitask learning \cite{radford2019language}, while GPT-3 introduces in-context learning (ICL), allowing adaptation solely through contextual inputs \cite{NEURIPS2020_1457c0d6}. Unlike meta-learning, ICL does not require parameter updates, making it well-suited for wireless applications. 
Theoretical studies reveal that Transformers approximate regression tasks \cite{garg2022can} and perform Bayesian inference \cite{panwarcontext}. Pre-training on diverse tasks enhances ICL performance \cite{raventos2024pretraining}, and Transformers exhibit implicit gradient descent during ICL, optimizing predictions without updating parameters \cite{von2023transformers, ahn2023transformers}. Research also highlights the importance of input structure, showing that distribution, label space, and structural cues significantly influence ICL performance \cite{chan2022data}.

ICL has been applied to wireless tasks such as channel estimation and resource allocation. \cite{teja2023transformers} first introduced Transformers for symbol detection, demonstrating their effectiveness as in-context estimators. \cite{zecchin2024context} extended this to MIMO systems, highlighting threshold behavior in ICL with increasing pre-training tasks, while \cite{zecchin2024cell} explored multi-user cell-free MIMO systems. LLMs have also been adapted for wireless applications, including resource management and network configuration. \cite{abbas2024leveraging} applies ICL-based LLMs for symbol detection, while \cite{shao2024wirelessllm, tong2024wirelessagent} propose WirelessLLM and WirelessAgent frameworks, respectively, leveraging knowledge alignment and multimodal data fusion. Unlike these approaches, our work focuses on training a Transformer-based receiver from scratch. 
Additionally, Transformers have also been used for radio fingerprinting \cite{ott2024radio}, automatic modulation recognition (AMR) \cite{rashvand2024enhancing}, spectrum sensing \cite{zhang2024spectrum}, and neuromorphic designs aimed at reducing power consumption \cite{song2024neuromorphic}. 


\section{System Model and Canonical Methods}
\label{sec:system model}
\subsection{Wireless Model}
\label{sec:problem}
To more clearly illustrate our design, we consider a canonical receiver symbol detection problem over a standard narrowband wireless fading channel. Broadly, we consider an {$(N_t,N_r)$} MIMO system, where the channel is represented by an $N_r \times N_t$ complex-valued matrix $H_t$ at time $t$, following a distribution $P_H$. We {do not specify the channel distribution but} normalize the channel coefficients such that each entry in $H_t$ has a unit variance. The received signal at time $t$ is expressed as: 
\begin{equation*}
    y_t = H_t x_t + z_t,
\end{equation*}
where the channel noise $z_t \in \mathbb{C}^{N_r}$ is modeled as a complex additive white Gaussian noise vector with zero mean and covariance matrix $\sigma^2 I$. Each entry of the input vector $x_t \in \mathbb{C}^{N_t}$ is sampled uniformly at random from a constellation set $\mathcal{X}$ (e.g., 16QAM), and this modulation process is independently and identically distributed (i.i.d.) across both time and space.  We normalize the signal to ensure a unit average total transmit power, i.e., $\mathbb{E}[\|x_t\|^2] = 1$. The average signal-to-noise ratio (SNR) at any receive antenna is thus given by $\text{SNR} = {1}/{\sigma^2}$.

We focus on the \emph{block-fading} channel model \cite{TV:05} as illustrated in \cref{fig:blockfading}. In this model, the channel $H_t$ remains constant over a coherent time of $T$ time slots, and is i.i.d. across different coherence periods. In other words, 
$H_t=h_l, \forall t=(l-1)T+1,\cdots, lT$, 
for the $l$-th coherence period where $h_l$ is drawn i.i.d. from $P_H$.  Correspondingly, the data transmission is organized into frames, where each frame has a length that is at most $T$. The frame structure is designed such that the first $k$ transmitted symbols are known and pre-determined \emph{pilot symbols}, whose original purposes include assisting the receiver to perform channel estimation \cite{Cav:91} of $h_l$ so that it can perform coherent symbol detection. Based on the reception of pilot pairs $D_{k} = \{(y_1, x_1), \cdots, (y_k, x_k)\}$, 
the goal is to design a demodulator that accurately recovers the transmitted symbols $x_{k+1}, \cdots, x_{T}$ from the received signals $y_{k+1}, \cdots, y_{T}$ with high probability and reasonable complexity.

\begin{figure}[!htbp]
    \centering
    \includegraphics[width=0.4\textwidth]{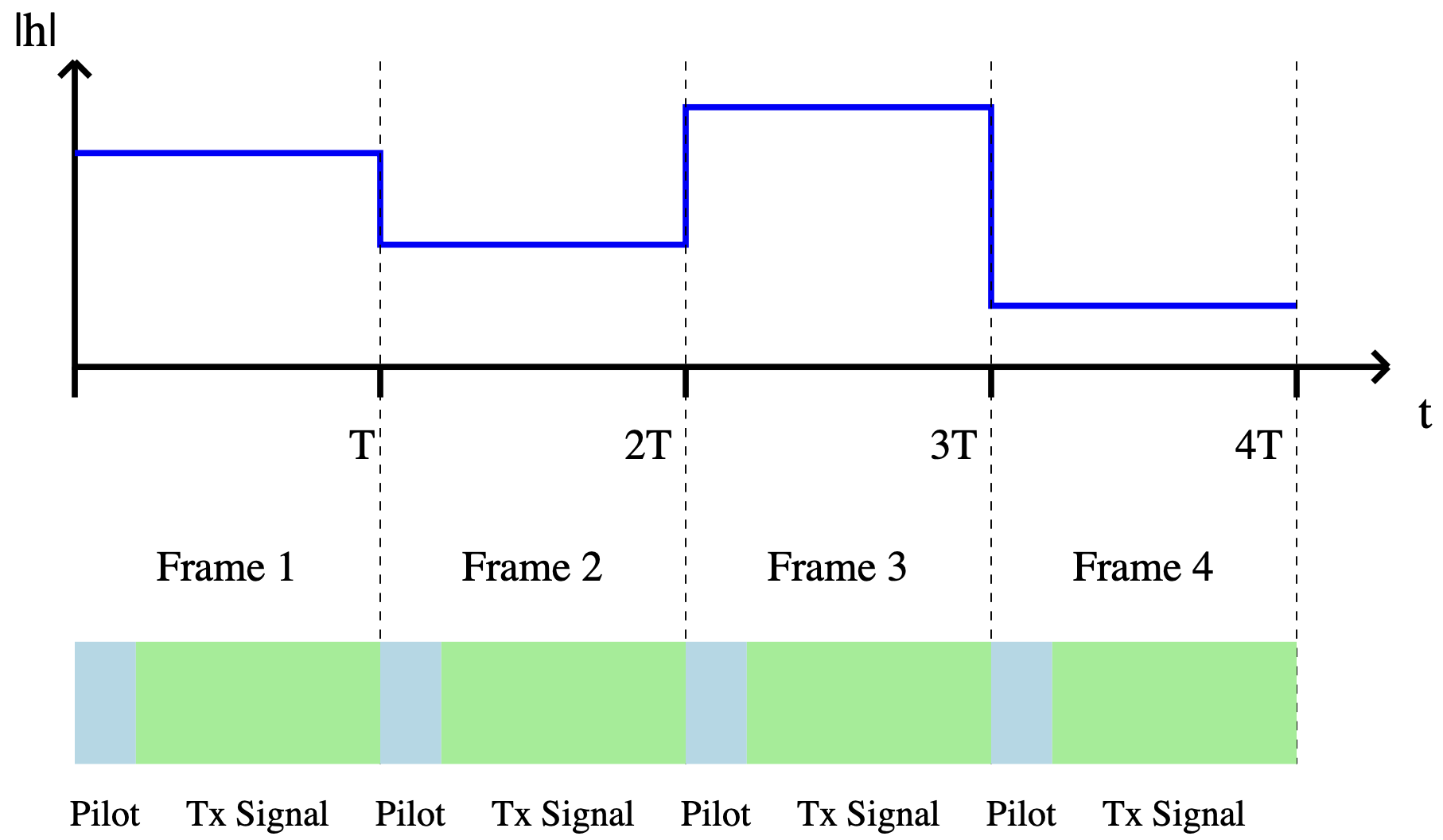}
    \caption{Illustration of block-fading channels and the associated frame structure. In each frame, the channel remains constant and the frame is divided into pilot signals and transmit data signals.} 
    \label{fig:blockfading}
\end{figure}

\subsection{Canonical Methods for Coherent Symbol Detection}
In the traditional approach, the receiver first estimates the channel using pilot signals, then performs symbol detection on the received signal $y_{t}$ via hypothesis testing for each $t = k+1, \cdots, T$. Typically, the (Linear) MMSE estimator is used for channel estimation, {which} is given by
$\hat{H} = Y X^{\dagger} (X X^{\dagger} + \sigma^2 I)^{-1}$, 
where $X = [x_1,x_2,\cdots,x_{k}] \in \mathbb{C}^{N_t \times k} $ is the pilot matrix and $Y = [y_1, y_2,\cdots,y_k] \in \mathbb{C}^{N_r \times k }$ is the received signal matrix. Then, with the estimated channel, the transmitted symbol $\hat{x}_{t}$ is detected by projecting $y_{t}$ onto the closest symbol in the modulation constellation $\mathcal{X}$ as
$\hat{x}_{t} = \arg \min_{x \in \mathcal{X}} \| \hat{H}x - y_{t} \|^2,  \forall t = k+1, \cdots, T$.

This two-step process treats channel estimation and symbol detection as separate tasks. Such decoupling can result in suboptimal detection, particularly under noisy conditions or limited pilot data \cite{Shen2008jsac}. Optimal estimators like MMSE rely on precise statistical models of the channel and noise, which can be difficult to obtain for complex wireless environments. These estimators are also computationally intensive due to matrix inversions and posterior probability calculations, making them less appealing for real-time high-dimensional systems.

To address these challenges, data-driven methods \cite{simeone2018very} have been explored for joint channel estimation and symbol detection. Deep learning architectures have shown promise \cite{le2021deep, neumann2018learning, lu2018mimo, aoudia2021end, park2020learning, chen2018neural}, but their high data requirements \cite{kim2020massive} and poor adaptability to varying channel conditions without retraining limit their real-world applicability \cite{simeone2020learning}.



\section{In-Context Learning-Based Symbol Detection}
\label{sec:ICL_detection}
ICL for symbol detection leverages the structure of wireless communication frames, particularly in block-fading channels where channel conditions remain stable during the coherence time. Within each frame, pilot signals are followed by subsequent received signals, which naturally align with the Transformer architecture's strength in processing \emph{sequence-based inputs}. The Transformer's ability to model dependencies among sequential data allows it to capture complex relationships within transmitted signals, making it highly effective for symbol detection tasks.

In this section, we discuss the ICL-based symbol detection method. We first formulate the symbol detection problem, and then present the ICL implementation using a popular Transformer model GPT-2, which also serves as the backbone of our proposed solution\footnote{We use GPT-2 with elaborated design choices as a concrete example throughout the paper. However, the proposed principle can be easily adapted to other Transformer architectures.}.

\subsection{ICL-based Problem Reformulation}
Each ICL symbol detection task $\tau$ corresponds to a latent channel $H$ and a channel noise level $\sigma^2$, drawn from the unknown joint distribution $P_{\tau} = P_H P_{\sigma^2}$. 
The receiver does not have any prior knowledge of the specific task $\tau$, meaning it does not know the current channel realization $H$ or the SNR level {$1/\sigma^2$}. Instead, it is provided only with a prompt $S_{t}^{\tau} = (D_{k}^{\tau}, y_{t})$,
consisting of $k$ target pairs $D_{k}^{\tau} = \{(y_1, x_1), \cdots, (y_k, x_k)\}$, 
which are sampled from the conditional distribution $P_{x,y|\tau}$ and serve as the in-context examples for the current task $\tau$, along with the query $y_{t}$, $\forall t = k+1, \cdots, T$. For block-fading channels, the context set $D_{k}^{\tau}$ corresponds to the \emph{pilot signals}. 

We note that a significant advantage of this reformulation (and the subsequently proposed solution) is that there is no need to change the {existing} frame structure or the design of pilot signals. Rather, the innovation is entirely at the receiver side where we leverage the pilot and decoded signals in a different way. This is an important advantage in practice as it allows for backward compatibility with the existing standard.

The goal of symbol detection is to identify the corresponding input signal $x_{t}$ for the new query signal $y_{t}$ from the same task due to the block-fading nature. The ICL-based detection makes its decision as $\hat{x}_{t} = f_{\theta}(S_{t}^{\tau})$, 
where $\theta$ represents the parameters of the model. The detection for the query is measured by the Symbol Error Rate (SER), which is the frequency at which transmitted symbols are incorrectly decoded. The expected SER for the new query with $k$ contexts, taking the expectation over the task distribution for $\forall k = 1, \cdots, T-1$, is defined as:
\begin{equation}
\text{SER}_{k}(\theta) = \mathbb{E}_{\tau}  \mathbb{E}_{D_k^{\tau},y_t|\tau} \left[ f_{\theta}(D_{k}^{\tau}, y_t) \neq x_{t} \right].
\label{eqn:expected_SER}
\end{equation}

\subsection{Vanilla In-Context Symbol Detection}
\label{sec:vanilla_ICL}

Transformer models have emerged as powerful tools for classification tasks \cite{shen2024training}, leveraging their ability to capture long-range dependencies. This approach to wireless symbol detection was introduced in \cite{teja2023transformers, zecchin2024context}. 
The input-output structure is illustrated in \cref{fig:model_ICL}. With a causal masked self-attention mechanism, the model outputs the detection $\hat{x}_{t}$ at the corresponding position of $y_{t}$, relying only on known preceding contexts and the received signal. During the forward process, the Transformer solves $k+1$ detection problems for the same task $\tau$, using increasing pilot data points. The results {in \cite{teja2023transformers, zecchin2024context}} demonstrate that the Transformer exhibits strong capabilities in symbol estimation within context, without requiring explicit model updates.

\begin{figure}[!htbp] 
    \centering
    \includegraphics[width=0.4\textwidth]{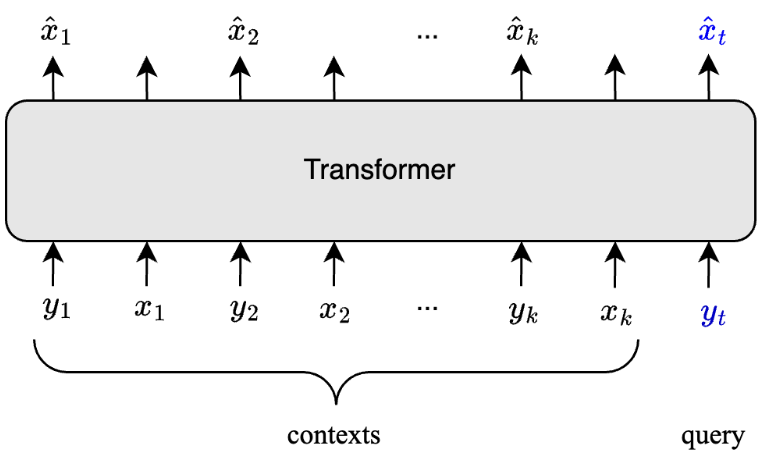}
    \caption{Decoder-only Transformer architecture for {ICL}-based symbol detection with $k$ pilots. Detection output applies to $t=k+1, \cdots, T$. }
    \label{fig:model_ICL}
\end{figure}

\section{Decision Feedback In-Context Detection}
\label{sec:DEFINED}
The vanilla ICL approaches for symbol detection require sufficient context to achieve accurate estimation, which is often impractical in real-world scenarios. Pilot signals are costly and limited, reducing their adaptability for practical applications. 
Furthermore, transmitting more pilot signals reduces system throughput and energy efficiency. 
For situations where the number of pilots is small, neither conventional two-stage (channel estimation then symbol detection) nor vanilla ICL solutions can achieve good performance. Additionally, the ICL approaches generally formulate the symbol detection task as a \emph{regression} problem, as in \cite{raventos2024pretraining, zecchin2024context, zecchin2024cell}, where a Transformer model is trained to minimize the {MSE} loss. Although their models achieve performances comparable to the optimal MMSE estimator for $x$, an additional projection step is required to map the output to the transmitted symbol, leading to a mismatch and losing optimality in the process.

In contrast, we directly formulate the problem as a \emph{multi-class classification} task \cite{shen2024training}, allowing the model to perform symbol detection while directly optimizing for SER during inference. This approach eliminates the need for explicit channel estimation, as the model implicitly gains knowledge of the channel through the pilot signal pairs. 
Additionally, inspired by the decision-feedback concept in wireless communication, we propose the \underline{DE}cision \underline{F}eedback \underline{IN}-Cont\underline{E}xt \underline{D}etection (\alg) {method} for symbol detection, as illustrated in \cref{fig:model_DFE}. In this approach, \alg sequentially feeds back the already decoded symbol pairs as \emph{noisy pilots} and incorporates them as part of the prompt. 
As will be seen, \alg achieves robust performance even under challenging conditions with only a \emph{single} pilot, outperforming previous ICL models that struggle with insufficient pilot data. At the same time, it maintains high accuracy when sufficient pilot data is available.

\begin{figure}[!htbp]
    \centering
    \includegraphics[width=0.4\textwidth]{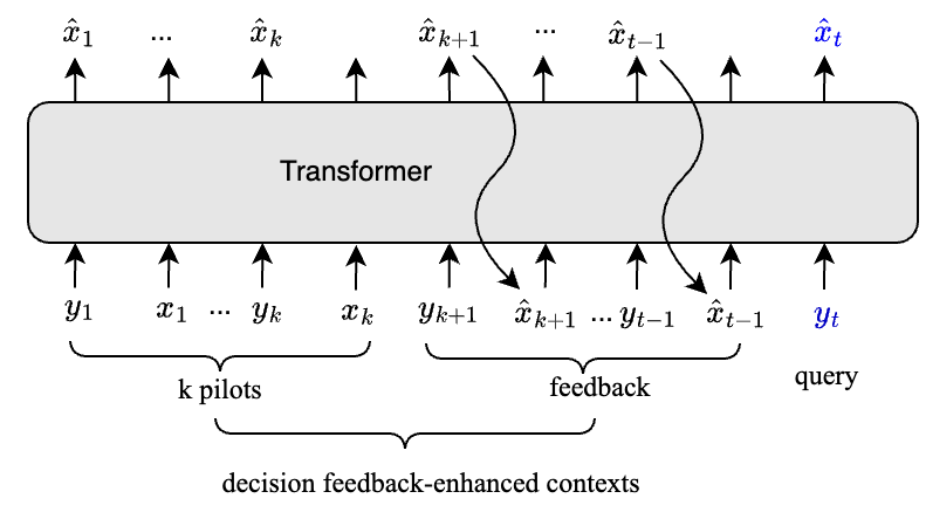}
    \caption{\alg model architecture with $k$ pilots and $(t-k-1)$ decision feedback contexts to detect $x_{t}$, $\forall t=k+1, \cdots, T$.}
    \label{fig:model_DFE}
\end{figure}


Next, we provide a detailed description of the \alg model structure and explain its data processing flow. Specifically, we highlight the critical role of the multi-head self-attention mechanism in enabling efficient information extraction. Furthermore, we discuss the model's size and describe the training process, which comprises dedicated pre-training and fine-tuning phases, as outlined in the following subsections.

\subsection{Tokenization, Attention, and Detection}

The overall model architecture of \alg is given in \cref{fig:model_attention}.  
The computation involves several key steps. 

First, tokenization converts the input IQ samples into embeddings, which are then processed through \(L\) cascaded layers of the decoder backbone that is adapted from the GPT-2 model. The model employs masked multi-head self-attention, a pivotal operation in Transformer architectures. This attention mechanism enables the output embeddings to capture long-range dependencies and extract relevant information from the input sequence. In the context of symbol detection, these dependencies represent the relationships between earlier pilot pairs and the current detection task involving the received query signal, both of which depend on the same channel $h_l$ in the block fading model. Capturing these relationships is critical for learning signal characteristics and the underlying channel structure, ultimately improving detection accuracy.

Following the self-attention step, the Transformer’s output passes through a classification head consisting of a linear transformation, which projects it into the classification dimension. Finally, a softmax layer computes the class probabilities, and the class with the highest probability is selected as the detected symbol. Each component is explained in detail below.

\begin{figure}[!htbp]
    \centering
    \includegraphics[width=0.35\textwidth]{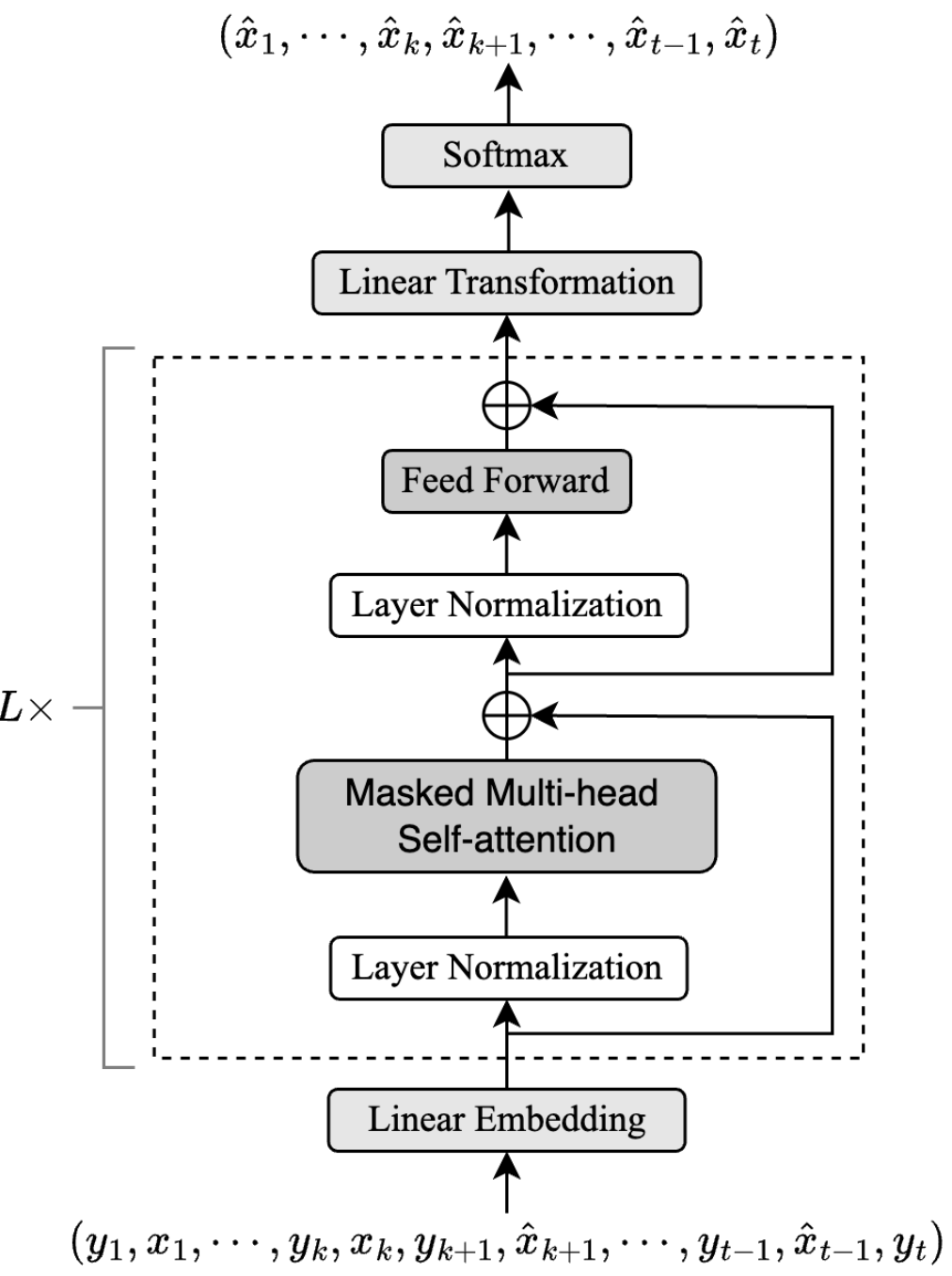}
    \caption{\alg model architecture featuring linear embedding for tokenization, \(L\) cascaded decoder layers, and a classification head.}
    \label{fig:model_attention}
\end{figure}

\subsubsection{Linear Embedding}
The first step in the data processing pipeline is tokenization, which converts the IQ {samples} into input embeddings. Both the transmitted and received signals are complex-valued. For the received signal $y \in \mathbb{C}^{N_r}$, we concatenate its real and imaginary parts as 
$\tilde{y} = [\Re(y), \Im(y)] \in \mathbb{R}^{2N_r}$, 
where $\Re(y)$ and $\Im(y)$ represent the real and imaginary parts of a complex signal $y$, respectively. 
For the transmitted signal $x \in \mathbb{C}^{N_t}$, where each symbol is sampled from a discrete constellation set with {size} $C$, we represent each signal using one-hot encoding: 
$\tilde{x} = \text{Onehot}(x) \in \mathbb{R}^{C^{N_t}}.$ 

Next, we apply zero padding to ensure uniform dimensionality $D_s = \max(2N_r, C^{N_t})$. The input is then passed through a linear embedding layer parameterized by a learnable matrix $A \in \mathbb{R}^{D_e \times D_s}$. The resulting embedding sequence is
$E = [A\tilde{y}_1, A\tilde{x}_1, \cdots, A\tilde{y}_{t}],$ 
which converts the input signals into a unified embedding space for further processing by the Transformer backbone.

\subsubsection{Decoder Layers with Masked Multi-Head Attention}  
Our model adopts the GPT-2 backbone with $L$ stacked decoder layers, where each layer refines the representation of the input sequence:  
$E^{(0)} = E, E^{(\ell + 1)} = \text{DecoderLayer}(E^{(\ell)}), \ell = 0, \dots, L-1.$ 
Each decoder layer consists of masked multi-head attention (MHA), a feedforward network (FFN), and residual connections with layer normalization. GPT-2 employs a pre-normalization scheme, where layer normalization is applied before each sub-layer \cite{radford2019language}. The forward pass computation of the Transformer follows a standard process, which we summarize here for completeness.

\paragraph{Masked Multi-Head Attention}  
MHA captures long-range dependencies across the input sequence, which, in our case, is established via the common channel $h_l$. Given the normalized input:
$\tilde{E}^{(\ell)} = \text{LayerNorm}(E^{(\ell)}),$ 
queries ($Q$), keys ($K$), and values ($V$) are computed as:
$Q = \tilde{E}^{(\ell)} W_Q, K = \tilde{E}^{(\ell)} W_K, V = \tilde{E}^{(\ell)} W_V.$ 
To enforce autoregressive behavior, a causal mask $M$ prevents attention to future tokens:
$\text{MaskedAttention}(Q, K, V) = \text{softmax}\left(\frac{Q K^\top}{\sqrt{d_k}} + M\right) V.$ 
Outputs from multiple attention heads are concatenated and linearly projected:
$\text{MHA}(\tilde{E}^{(\ell)}) = \text{Concat}(\text{head}_1, \dots, \text{head}_h) W_O.$ 
A residual connection is then applied:
$U^{(\ell)} = E^{(\ell)} + \text{MHA}(\tilde{E}^{(\ell)}).$ 

\paragraph{Feedforward Network}  
The FFN processes the intermediate output through two linear transformations with a ReLU activation: 
$\text{FFN}(\tilde{U}^{(\ell)}) = \text{ReLU}(\tilde{U}^{(\ell)} W_1 + b_1) W_2 + b_2.$ 
A residual connection produces the final layer output: 
$E^{(\ell + 1)} = U^{(\ell)} + \text{FFN}(\text{LayerNorm}(U^{(\ell)})).$ 
Pre-normalization stabilizes training, as shown in \cite{radford2019language}.

\subsubsection{Classification Layer}  
The final decoder output, $E^{(L)}$, is projected onto the classification space, where a softmax function determines the most likely transmitted symbol: 
$\hat{x} = \arg \max \left(\text{softmax}(W_c E^{(L)})\right),$ 
where $W_c$ is the classification weight matrix.

\subsection{Model Parameters}

Our specific Transformer model is designed with an embedding dimension of $d_e = 64$, $L = 8$ layers, and $h = 8$ attention heads, resulting in approximately \textbf{0.42 million} parameters, which is significantly smaller compared to large language models (LLMs) commonly applied in wireless communication tasks, such as those discussed by \cite{shao2024wirelessllm, abbas2024leveraging}. For instance, even the smallest LLM used in \cite{abbas2024leveraging}, GPT-J 6B, contains over 6 billion parameters, making it approximately 14,000 times larger than \alg. 
The compact size of \alg not only enables its deployment on mobile devices but also significantly shortens the inference time, enabling low-latency detection at the receiver. 

\subsection{Training Design}

In this section, we describe our data generation process and the training of \alg. 

\subsubsection{Data Generation}
\label{sec:data_generation}
We generate data according to the wireless system model described in \cref{sec:problem}. Specifically, we consider both SISO and $2 \times 2$ MIMO systems and explore various modulation schemes, including BPSK, QPSK, 16QAM, and 64QAM. For each wireless system and modulation task, we generate prompts consisting of sequences with $T$ pairs, where the maximum sequence length is set to $T = 31$. The training data only use Rayleigh block-fading channels, with the channel coefficient sampled as $H \sim P_H = \mathcal{CN}(0, 1)$ and remaining constant over each block. The channel noise is sampled i.i.d. from a Gaussian distribution $z_t \sim \mathcal{CN}(0, \sigma^2 I)$. Within each block, the noise variance $\sigma^2$ is constant and shared across all noise samples. However, for different blocks, $\sigma^2$ is i.i.d. sampled from a uniform distribution $P_{\sigma^2}$ over the range $[\sigma_{\min}^2, \sigma_{\max}^2]$.
Varying noise levels during training enhances generalization by exposing the model to diverse channel conditions. 
This variability applies only during training, while evaluation is conducted at fixed SNR levels for controlled performance assessment.
The training batch size is set to 512.

\subsubsection{ICL Pre-training} 

\begin{figure}[!htbp]
    \centering
    \includegraphics[width=0.5\textwidth]{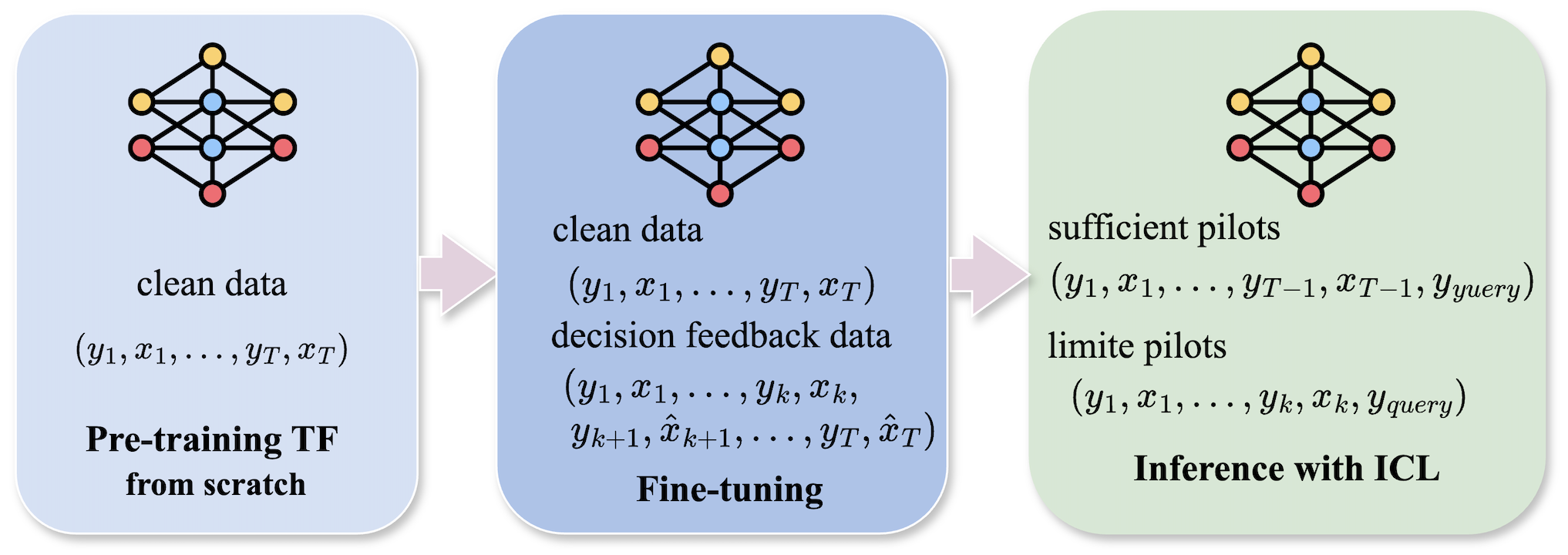}
    \caption{ The training process includes pre-training on clean data, followed by fine-tuning on a mixed dataset of clean and decision feedback noisy data. The model demonstrates strong performance, adapting to both limited and sufficient pilot scenarios during symbol detection (i.e., inference).}
    \label{fig:train_process}
\end{figure}

We delve into the details of \alg pre-training using the previously described data generation, which is highly nontrivial. To prepare for the discussion, we define \textbf{ICL-training} and \textbf{ICL-testing} as {the corresponding} operations on ground-truth data, represented by the clean prompt for $t = 1, 2, \ldots, T$: $S_{t}^{\text{ICL}} = \{y_1, x_1, \ldots, y_{t-1}, x_{t-1}, y_t\}$.  
On the other hand, \textbf{DF-training} and \textbf{DF-testing} utilize iteratively decoded sequences with $k$ pilot data and model decision feedback, operating on the decision feedback prompt for $t = k+1, \ldots, T$: $S^{\text{DF}}_{t} = \{y_1, x_1, \ldots, y_k, x_k, y_{k+1}, \hat{x}_{k+1}, \ldots, y_{t-1}, \hat{x}_{t-1}, y_t\}$. 
Each estimation $\hat{x}_j$ for $j = k+1, \cdots, t-1$ is the model output, which depends on the first $k$ pilot points and preceding model decisions. This output evolves during model training and gradually converges as training progresses.

We next define the loss functions for ICL-training and DF-training. The adopted loss function is the cross-entropy loss between the model’s output and the ground-truth labels, defined as $\text{loss}(\hat{x}, x) = -\sum_{c \in \mathcal{X}} \mathbbm{1}(x = c) \log P(\hat{x} = c)$, 
where $\hat{x}$ is the predicted output, $x$ is the ground-truth symbol, $\mathcal{X}$ represents the modulation constellation set, and $\mathbbm{1}(\cdot)$ is the indicator function.

Using this loss definition, the loss functions for ICL-training and DF-training are formulated as:
\begin{align}
    \mathcal{L}^{\text{ICL}}(\theta) &=  \frac{1}{N T} \sum_{i=1}^{N} \sum_{t=1}^{T}  \text{loss} \left( f_{\theta}(S_{t,i}^{ICL}), x_{t,i} \right), \label{loss:ICL} \\
    \mathcal{L}^{\text{DF}}(\theta) &=  \frac{1}{N (T-k)} \sum_{i=1}^{N} \sum_{t=k+1}^{T} \text{loss} \left( f_{\theta}(S_{t,i}^{DF}), x_{t,i} \right), 
    \label{loss:DF}
\end{align}
respectively, where $\theta$ represents the model parameters, and $N$ is the number of samples.

In DF-training, we first freeze the Transformer model's gradients to prevent parameter updates during the generation of decision feedback prompts. This process involves iterative forward passes, where the model's noisy output is incorporated into the prompt to create a new context. The generated prompts are then used to train the model and update its parameters. However, DF-training is time-intensive, as each step of detection and feedback requires a forward pass through the model. Additionally, the presence of noisy data complicates convergence, making effective training more challenging.

To tackle this computational challenge, one might consider directly employing ICL training for \alg. However, this approach would lead to a mismatch: \alg is trained on clean data but evaluated with noisy contexts. This discrepancy will naturally lead to performance degradation, even though we have empirically observed that ICL-training is approximately ten times faster than DF-training, as it eliminates the need for iterative data sampling and exclusively uses clean data.

Considering all factors, our final proposed solution is to \emph{perform ICL-training first, followed by tailored DF-training}. Here, ICL-training serves as pre-training, while tailored DF-training acts as fine-tuning, similar to the pre-training and fine-tuning processes used in LLMs. Training epochs are carefully structured into two phases, as conceptually illustrated in \cref{fig:train_process} and empirically shown in \cref{fig:train_loss}. In the first phase, the model converges just before reaching a plateau, at which point we transition to the tailored DF-training method. During this transition, a spike in the training loss is observed due to the shift in the training data distribution.  As shown later, ICL pre-training not only accelerates convergence but also improves recognition of clean data and ICL-testing performance.

\begin{figure}[!htbp]
    \centering
    \includegraphics[width=0.45\textwidth]{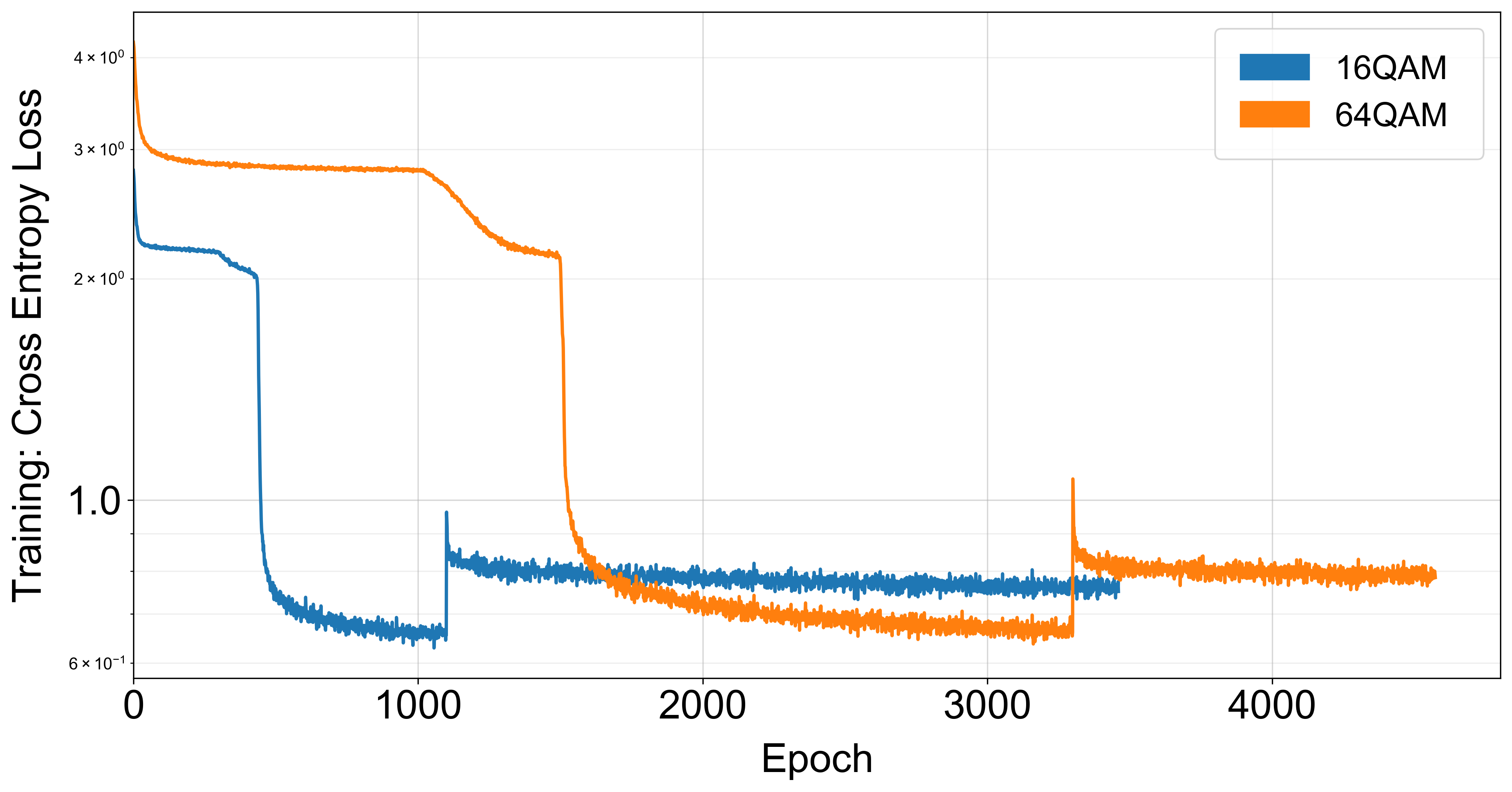}
    \caption{Training process loss evolution. The Transformer is initially trained from scratch using ICL pre-training, followed by tailored DF fine-tuning. The training processes are shown for SISO 16QAM and 64QAM, respectively. In the first phase, the loss exhibits an initial lull where it remains constant for a period, which is typical in ICL training, before decreasing sharply. Prior to reaching a plateau, the transition to tailored DF fine-tuning introduces a spike in loss due to the data distribution shift, followed by a gradual decrease as the model converges.}
    \label{fig:train_loss}
\end{figure}

\subsubsection{Decision Feedback Fine-tuning}
Instead of the vanilla DF-training in the second phase, we employ a carefully designed fine-tuning process. The new loss function is formulated as a \emph{weighted combination} of the previously defined losses in \cref{loss:ICL,loss:DF}, where $\alpha$ controls the balance between them: 
\begin{equation}
\mathcal{L}^{\text{fine-tuning}}(\theta) = \alpha \mathcal{L}^{DF}(\theta) + (1 - \alpha) \mathcal{L}^{ICL}(\theta).
\label{eqn:tunningloss}
\end{equation}
We set hyperparameter $\alpha$ to 0.7 in the experiment. 


Training on both clean and noisy data enhances the robustness of the Transformer model by exposing it to a more diverse dataset. Ultimately, we propose that \textit{a single Transformer model} can be trained to perform both ICL-testing and DF-testing, making our \alg model adaptable for practical wireless systems. For example, in scenarios with sufficient pilot information, the model can operate in the ICL manner. However, in challenging situations -- common in real-world applications -- where pilots are limited and difficult to acquire, the model can utilize previous decisions to improve performance in subsequent symbol detection. {In fact, as we will see in the experimental results, the pre-trained \alg model can handle previously unseen channel fading distribution very well. Remarkably, all of these generalization properties can be achieved with a single pre-trained Transformer.}

\subsection{Curriculum Learning} 

\begin{figure}[!htbp]
    \centering
     \includegraphics[width=0.45\textwidth]{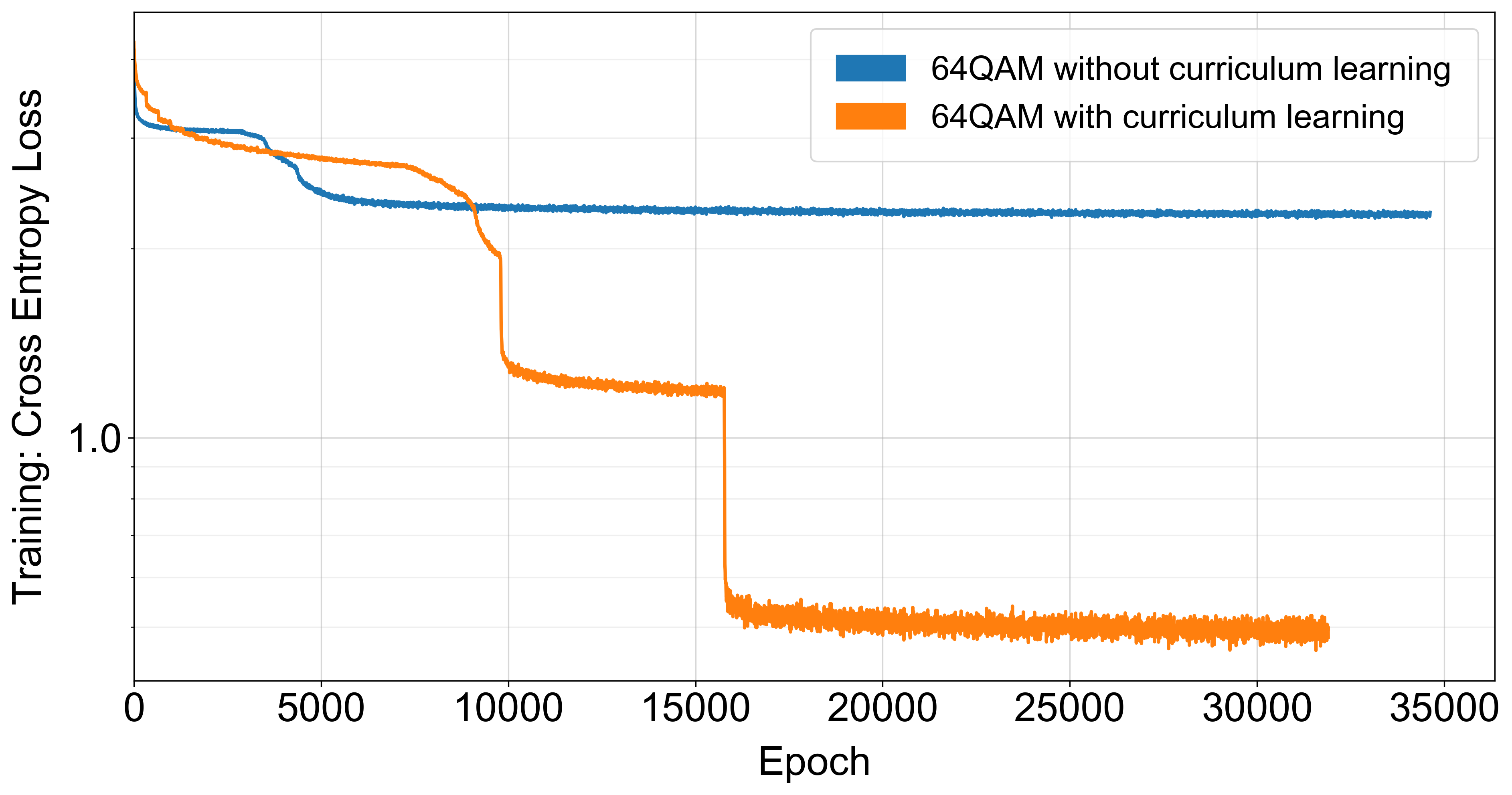} 
    \caption{Training cross-entropy loss for 64QAM in SISO. The model trained without curriculum learning fails to converge, even after an extended lull period. In contrast, the model trained with curriculum learning exhibits a significantly shorter lull period, leading to faster and more stable convergence.} 
    \label{fig:train_Curriculum}
\end{figure}

We further introduce curriculum learning to accelerate the training speed \cite{garg2022can}. The core idea is to gradually increase the task difficulty, similar to a learning curriculum, progressing from simple to more challenging scenarios. In \cite{garg2022can}, this technique was applied to a linear regression problem by gradually increasing the underlying dimension and context length until the target dimension and context length were reached. The authors observed that Transformer training typically experiences an initial lull period, where the loss remains stagnant before eventually dropping sharply. Curriculum learning mitigates this lull period and accelerates convergence.

In our wireless experiments, while the channel dimension remains fixed, we progressively increase the context length every fixed number of epochs until reaching the target length. As shown in \cref{fig:train_Curriculum}, this approach accelerates training by reducing the lull period. Notably, for the challenging 64QAM detection task with limited pilot data and low SNR, the model may fail to converge without curriculum learning.

\section{Theoretical Analysis}
\label{sec:theory}





In this section, we present a theoretical lower bound on the detection error and analyze the generalization behavior of \alg under channel distribution mismatch, particularly when the SNR and Line of Sight (LoS) components differ between training and testing. Our results build upon the techniques developed in \cite{shen2024training}. To simplify the analysis and provide cleaner insight, we focus on the BPSK SISO case, but the results can be extended to higher modulation levels. All proofs can be found in the Appendix.

To ensure a self-contained analysis, we first introduce the necessary definitions and assumptions adapted from \cite{shen2024training}. 
We denote the sigmoid function by $S(x) := 1/(1 + \exp(-x))$. We say a data pair $(y, x) \sim \dis^b(\mu_0, \mu_1, \Lambda)$ if $x \in \{-1, 1\}$ with equal probability and
\[
f(y \mid x = -1) = \mathcal{N}(\mu_0, \Lambda), \quad f(y \mid x = 1) = \mathcal{N}(\mu_1, \Lambda),
\]
where $\mu_0, \mu_1 \in \mathbb{R}^d$, and $\Lambda \in \mathbb{R}^{d \times d}$ is a positive definite covariance matrix.

\begin{assumption}
\label{asp:test_prompt}
For an in-context test prompt $P_{\test} = (y_1, x_1, \dots, y_k, x_k, y_{\query})$, we assume: (1)  $\{(y_i, x_i)\}_{i=1}^k \iid \dis^b(\mu_0, \mu_1, \Lambda)$; (2) $\mu_0^\top \Lambda^{-1} \mu_0 = \mu_1^\top \Lambda^{-1} \mu_1$.
\end{assumption}

We note that \cref{asp:test_prompt} is satisfied for the wireless symbol detection problem. For example, in our BPSK SISO Rayleigh fading setting, the received $y_{\query}$ follows a distribution 
\[
(y, x) \sim \dis^b(\mu_0, \mu_1, \Lambda) = \dis^b(-1, +1, \sigma^2 I),
\]
where $\mu_0 = -1$, $\mu_1 = +1$, and $\Lambda = \sigma^2 I$. 

Following the same parameterization of the one-layer linear Transformer as in \cite{shen2024training}, the model output for a test prompt $P_{\test}$ with weight matrix $W$ is given by
\begin{align}\label{eq: output}
    \widehat{x}_{\out} = S\left( \left( \frac{1}{k} \sum_{i=1}^k \testx_i \testy_i^\top \right) W \testy_{\query} \right),
\end{align}
The scalar output $\widehat{x}_{\out} \in (0,1)$ represents the model’s soft prediction. The predicted label $\widehat{x}_{\query} \in \{-1, 1\}$ is then sampled according to:
$\probp{\widehat{x}_{\query} = 1} = \widehat{x}_{\out}, \probp{\widehat{x}_{\query} = -1} = 1 - \widehat{x}_{\out}.$

\subsection{Convergence}
Following the same training procedure described in Section 3.1 of \cite{shen2024training}, and according to its Theorem 3.3, {a single-layer linear Transformer with parameter $W$ trained via gradient descent and infinitely long prompts with covariance $\Lambda$ can converge to the optimal weight matrix $W^* = 2\Lambda^{-1}$.}
In the matched setting where the training and test distributions share the same covariance $\Lambda$, we have the following result. 

\begin{theorem}
\label{thm: lower bound}
Consider a test prompt $P_{\test} = (y_1, x_1, \dots, y_k, x_k, y_{\query})$ satisfying Assumption~\ref{asp:test_prompt}. 
Let $\widehat{x}_{\query}$ denote the predicted label of the Transformer with {optimal} parameter $W^* = 2\Lambda^{-1}$. Then,
\begin{align*}
&\E\left[ \left| \Prob(\widehat x_{\query}=1 \mid y_{\query}) - \Prob(x_{\query}=1 \mid y_{\query}) \right|^2 \right] \\
=& \frac{1}{k} \left( S'(\mu^\top \Lambda^{-1} q) \right)^2 
\left[\frac{(u^\top \Lambda^{-1} q)^2}{4}  + 4 q^\top \Lambda^{-1} q
\right] +\; o\left(\frac{1}{k}\right),
\end{align*}
where $\mu = \mu_1 - \mu_0$, $u = 2(\mu_1 + \mu_0)$, and $q = y_{\query}$. The expectation is taken over the in-context samples $\{(y_i, x_i)\}_{i=1}^k \iid \dis^b(\mu_0, \mu_1, \Lambda)$.
\end{theorem}

\begin{remark}
\label{remark1}
During testing, the prediction error of a well-trained Transformer is lower bounded at the order of $O(1/k)$. If the test prompt is also of infinite length, the Transformer can asymptotically achieve perfect prediction. Notably, Theorem~\ref{thm: lower bound} does not impose strong distributional assumptions on $\mu_0$ and $\mu_1$, which highlights the strong generalization ability of the trained Transformer.
\end{remark}

\begin{remark}
Since this lower bound is derived under the assumption of ground-truth in-context samples, it also serves as a lower bound for \alg, where the prompt consists of a limited number of pilot pairs and decision feedback (pseudo-labeled) symbols. In high-SNR scenarios, decision feedback is correct with high probability, enabling \alg to closely approach this theoretical bound. However, the presence of noisy feedback complicates the derivation of a tight upper bound. A rigorous theoretical analysis of Transformer-based in-context learning with decision feedback remains an open and compelling direction for future research.
\end{remark}

\subsection{Generalization}
In the following, we rigorously prove that a single-layer Transformer trained at a specific SNR level can be deployed in the same Rayleigh fading channel under different SNR conditions. This result demonstrates the inherent robustness of the Transformer, which we further validate through experimental results presented in the following sections.

\begin{theorem}
\label{thm:mismatch}
Suppose a test prompt $P_{\test} = (y_1, x_1, \dots, y_k, x_k, y_{\query})$ satisfies Assumption~\ref{asp:test_prompt} with covariance $\Lambda = \sigma^2 I$, i.e., $\{(y_i, x_i)\}_{i=1}^k \iid \dis^b(\mu_0, \mu_1, \sigma^2 I)$. 
Now suppose the Transformer is trained on data with covariance $\Lambda' = \xi^2 I$, resulting in a learned {optimal} weight matrix $W^* = 2\xi^{-2} I$. Then the Transformer output is
$\widehat{x}_{\out} = S\left( \xi^{-2} \left( \frac{2}{k} \sum_{i=1}^k \testy_i \testx_i^\top \right) \testy_{\query} \right),$
and the {argmax} predicted label is given by
$\widehat{x}_p =
\begin{cases}
1, & \text{if } \widehat{x}_{\out} > 1/2, \\
-1, & \text{otherwise}.
\end{cases}$
Then, as $k \to \infty$, we have that $\widehat{x}_p$ converges to the optimal {predictor that maximize the accuracy between $\widehat{x}_p$ and $x_\query$}.
\end{theorem}

\begin{remark}
\cref{thm:mismatch} demonstrates that a Transformer trained under one noise level ($\xi^2$) can generalize effectively to test-time noise levels with (1) different $\sigma^2$, and (2) different $\mu_0$, $\mu_1$. 
In a Rayleigh or Ricean fading channel, the optimal predictor that maximizes the accuracy between $\widehat{x}_p$ and $x_\query$ reduces to predicting based on the sign of $(\mu_1 - \mu_0)^\top x_\query$. Under SNR mismatch and LoS mismatch, only a multiplicative constant changes, while the form of the estimator remains unchanged. 
Furthermore, we will see in \cref{sec:exp_mismatch} that empirical experiments corroborate these theoretical findings. 
\end{remark}


\section{Experiment}

\begin{figure*}[!htbp]
    \centering
    \subfloat[\centering \footnotesize BPSK, SNR = 15 dB, 1 or 2 Pilots, $\text{P1 gain}_{\text{DF}} =22.8\%$.]{%
        \includegraphics[width=0.45\textwidth]{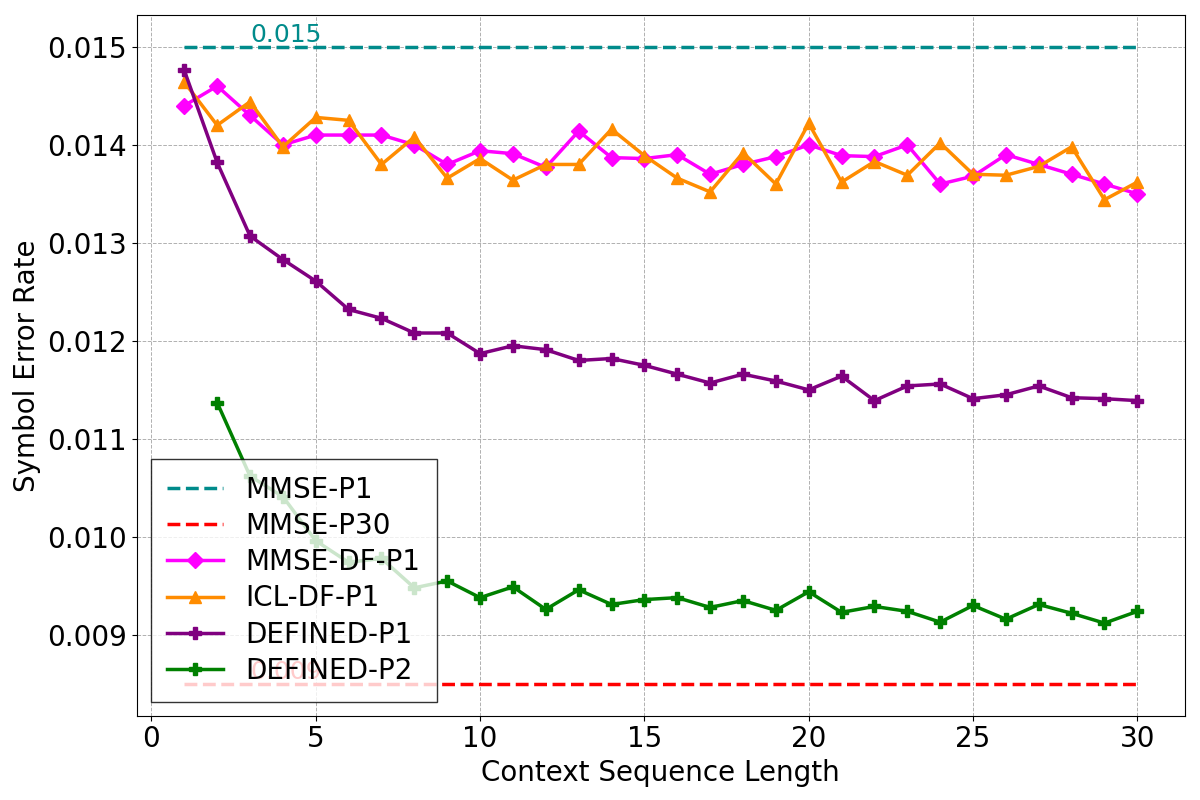}%
        \label{fig:SISO_BPSK_SNR15_P1}%
    }%
    \subfloat[\centering \footnotesize QPSK, SNR = 20 dB, 1 or 2 Pilots, $\text{P1 gain}_{\text{DF}} = 19.3\%$.]{%
        \includegraphics[width=0.45\textwidth]{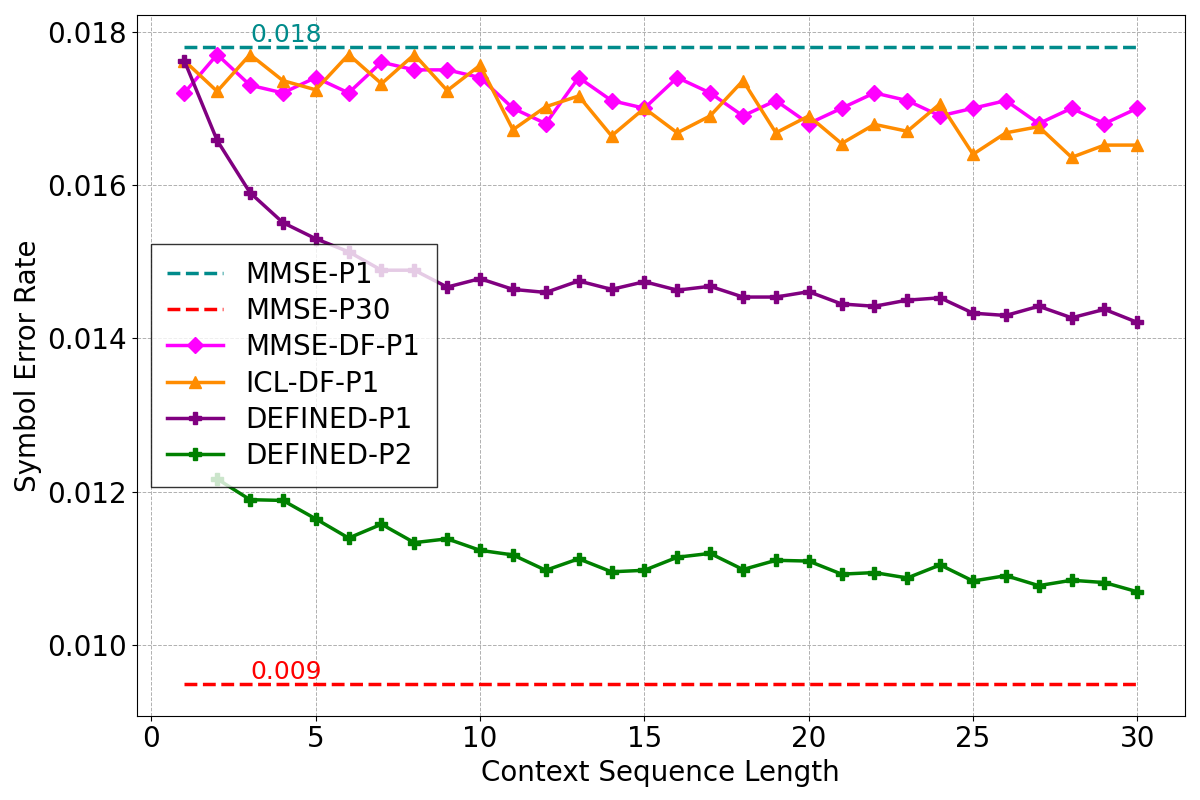}%
        \label{fig:SISO_QPSK_SNR20_P1}%
     }%
    \hfill
    \subfloat[\centering \footnotesize 16QAM, SNR = 30 dB, 1 or 2 Pilots, $\text{P1 gain}_{\text{DF}} = 55.3\%$.]{%
        \includegraphics[width=0.45\textwidth]{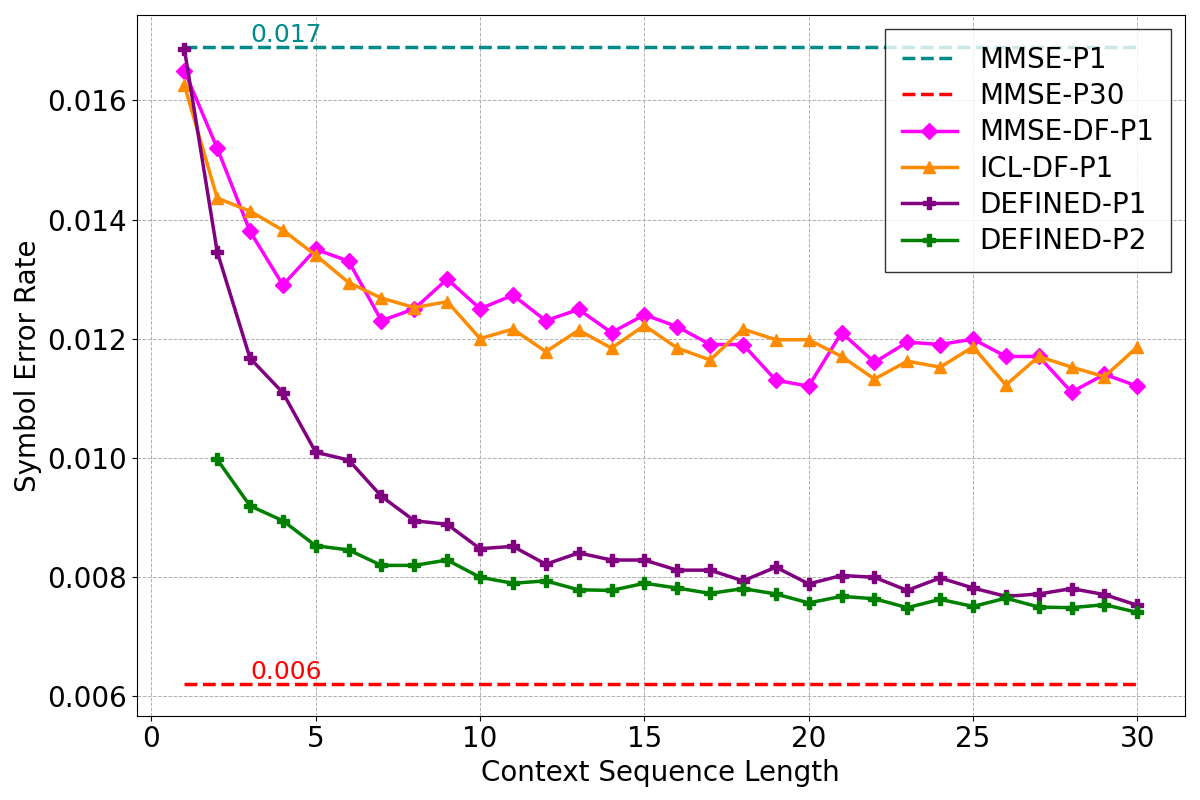}%
        \label{fig:SISO_16QAM_SNR30_P1}%
    }%
    \subfloat[\centering \footnotesize 64QAM, SNR = 35 dB, 1 or 2 Pilots, $\text{P1 gain}_{\text{DF}} = 62.6\%$.]{%
        \includegraphics[width=0.45\textwidth]{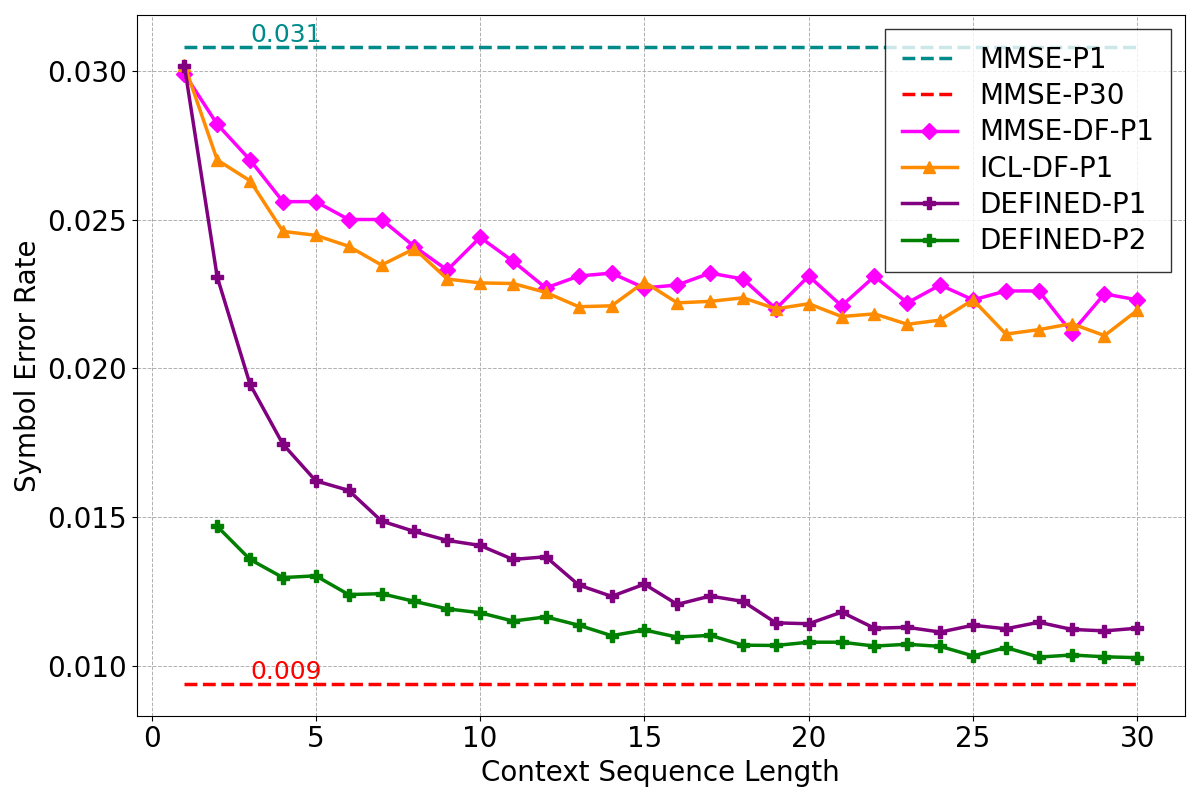}%
        \label{fig:SISO_64QAM_SNR35_P1}%
    }%
    \caption{SER as a function of context sequence length for SISO systems with BPSK, QPSK, 16QAM, and 64QAM under varying SNRs and pilot lengths.}
    \label{fig:DEFINE_SISO}
\end{figure*}

\label{sec:expriment}
We present the experimental results and analyze the performances of {\alg} in comparison with baseline algorithms. {The results suggest that} our model achieves superior performance not only with sufficient pilot data but also exhibits notable improvements in limited-data scenarios by effectively leveraging noisy feedback. Furthermore, \alg demonstrates robust performance in complex modulation tasks as well as mismatched channel distributions. 

\subsection{Baseline Algorithms}
We first describe several baseline algorithms against which \alg will be compared.

\subsubsection{In-Context Learning}
We train a Transformer from scratch using vanilla ICL-training and {evaluate the SER for both ICL-testing and DF-testing, shown as ``ICL-ICL'' and ``ICL-DF'' curves in the figures, respectively. Note that ICL-ICL does not utilize any decision feedback, while ICL-DF has a mismatch since the training phase does not have any decision-feedback data, but the testing phase utilizes such feedback.}

\subsubsection{MMSE Algorithm}
This is the canonical coherent detection algorithm, which first estimates the channel using pilots and then performs detections for the subsequent received signals using the estimated channel. Assuming a known pilot signal matrix $X = [x_1, x_2,\cdots, x_k]$, the received signal matrix is represented as $Y = HX+Z =  [y_1, y_2,\cdots, y_k]$.  {In the case of Rayleigh fading,} both the channel and noise follow complex Gaussian distributions, and the pair $(H, Y)$ is jointly Gaussian. The MMSE estimator for $H$ can be derived as: $\hat{H}_{k}^{\text{MMSE}} =  Y X^{\dagger} (X X^{\dagger} + \sigma^2 I)^{-1}$ for the i.i.d. fading case. 
Note that when the channel fading is not Gaussian, this is the linear MMSE (LMMSE) channel estimator. Then, with the $t$-th received signal $y_{t}$, the transmitted symbol $x_{t}$ is estimated by projection onto the closest symbol in $\mathcal{X}$:
\begin{equation}
\hat{x}_{t} = \arg\min_{x \in \mathcal{X}} \| \hat{H}_{k}^{\text{MMSE}} x - y_{t} \|^2, \forall t = k+1, \cdots, T.
\label{eqn:projection}
\end{equation}
With $k$ pilot pairs, the mean SER is computed and shown as a horizontal line labeled ``MMSE-P$k$''.

\subsubsection{MMSE-DF Algorithm}

MMSE-DF is an extension of the previous MMSE solution by sequentially using the decision feedback data as if they were new pilot pairs. This method belongs to decision-directed (or data-aided) channel estimation \cite{deng2003decision,karami2006decision}.  Starting with $k$ pilots, we compute the MMSE estimator of $H$ and detect $\hat{x}_{k+1}$ using $y_{k+1}$, as described in \cref{eqn:projection}. The decision pair $(y_{k+1}, \hat{x}_{k+1})$ is merged with the existing dataset and iteratively used to detect each subsequent signal until $\hat{x}_{T}$. The SER is plotted against the decision feedback-extended context sequence length.

\subsection{Experimental Results}
During testing, we randomly sample 80,000 prompts to compute the average SER across tasks. 
Channel realizations follow i.i.d. Rayleigh distributions in both training and testing (except in \Cref{sec:exp_mismatch}), with the difference being that in the training phase, the dataset has different SNRs as described in \cref{sec:data_generation}. Results are presented for BPSK, QPSK, 16QAM, and 64QAM in the SISO system, and for QPSK in a $2 \times 2$ MIMO system. We note that for methods that do not utilize feedback (such as ICL-ICL), the SER is plotted as a function of the \emph{pilot} sequence length to evaluate the impact of different numbers of pilots. For methods that utilize decision feedback (such as ICL-DF, MMSE-DF, and \alg), SER is plotted as a function of the \emph{context} sequence length. In this case, we emphasize that the adopted number of pilots is fixed, and we increase the context length by iteratively adding detected symbols. Specifically, the $t$-th point represents the SER of the estimator $\hat{x}_{t+1}$ using $t$ contexts, omitting the 0-th point. A glossary summarizing key training and testing configurations is provided in \cref{tab:glossary} for clarity.


To quantify the SER improvement with increasing context length under DF-testing, we define the gain metric as:
\begin{equation} \text{gain}_{\text{DF}} = \left( \frac{\text{SER}_{k}(\theta) - \text{SER}_{T-1}(\theta)}{\text{SER}_{k}(\theta)} \right) \times 100 \%, 
\end{equation}
which represents the percentage reduction in SER as the context length increases from $k$ to $(T-1)$, starting from $k$ known pilots.

\begin{table}[htbp]
    \centering
    \caption{Glossary of Training and Testing Methods}
    \label{tab:glossary}
    \begin{tabular}{ll}
        \toprule
        \textbf{Term} & \textbf{Description} \\
        \midrule
        \textbf{ICL-training} & Training with only clean pilot sequences. \\
        \textbf{DF-training} & Training with decision feedback sequences. \\
        \textbf{ICL-testing} & Testing with clean pilots only. \\
        \textbf{DF-testing} & Testing with decision feedback prompts. \\
        \midrule
        \textbf{Pilot Sequence Length} & Number of pilots. \\
        \textbf{Context Sequence Length} & Total sequence length, including feedback. \\
        \midrule
        \textbf{ICL (ICL-ICL)} & ICL-trained, ICL-tested (no feedback). \\
        \textbf{ICL-DF} & ICL-trained, DF-tested (mismatch). \\
        \textbf{DEFINED (DEFINED-DF)} & DF-trained, DF-tested. \\
        \textbf{DEFINED-ICL} & DF-trained, ICL-tested. \\
        \midrule
        \textbf{MMSE-P$k$} & MMSE with $k$ pilots (baseline). \\
        \textbf{MMSE-DF} & MMSE with decision feedback. \\
        \bottomrule
    \end{tabular}
\end{table}

\subsubsection{Comparison with Baseline Algorithms under DF-testing}

We first compare \alg with the baseline algorithms under DF-testing in a variety of scenarios, including some with the most \emph{extremely limited} pilot data: a \textbf{single pilot} in the SISO system. 
The results are presented in \cref{fig:DEFINE_SISO}. 


The \alg lines represent the performance of our proposed model during DF-testing, demonstrating a significant reduction in SER
as more decision feedback data is incorporated, i.e., increased context sequence lengths. 
For instance, in \cref{fig:SISO_16QAM_SNR30_P1}, for 16QAM with a single pilot, the SER decreases from 0.0168 to 0.0075 as the feedback context sequence length increases, achieving a $55.3\%$ SER gain. Remarkably, our model achieves performance comparable to conventional methods that require more than $4$ pilots while \alg using only a single pilot.
This confirms that \alg effectively utilizes noisy feedback to improve detection performance with limited pilot data. Furthermore, we note that increasing the amount of pilots from 1 to 2 would improve the detection accuracy for the early detected symbols due to more clean labels. As the context length increases, the performance gradually matches that of 1-pilot since the total (clean and pseudo) labels are roughly the same for both cases.

On the other hand, ICL-DF, which represents a Transformer trained with ICL method but deployed using decision feedback as if the detected symbols are pilots, performs significantly worse than \alg, nearly coinciding with MMSE-DF. This observation highlights that models trained solely on clean data struggle with noisy feedback, as they simply treat all feedback data as clean. This underscores the importance of DF fine-tuning to effectively handle noisy feedback.

\subsubsection{Comparison with Baseline Algorithms Without Decision Feedback}

We then compare \alg with baseline algorithms under ICL-testing (i.e., no DF). Results are reported in \cref{fig:ICL_results}.
The ICL-ICL result, representing the model trained and tested using the ICL method, demonstrates that with 30 pilots Transformer slightly outperforms the MMSE algorithm with 30 pilots during ICL testing. This improvement arises from the model’s ability to jointly perform channel estimation and symbol detection. 

Importantly, \alg also performs well in ICL testing, as indicated by the DEFINED-ICL result, which closely aligns with the ICL-ICL line. This suggests that ICL pre-training, followed by carefully designed loss functions during decision feedback fine-tuning, empowers the model to effectively learn from clean data. Together with the results from DF-testing, we conclude that \textbf{a single pre-trained Transformer can perform both DF and no-DF tasks}. Our \alg model adapts well to real-world symbol detection scenarios, excelling with ample pilot data and maintaining strong performance even with a single pilot.

\begin{figure}[!htbp]
    \centering
    \includegraphics[width=0.45\textwidth]{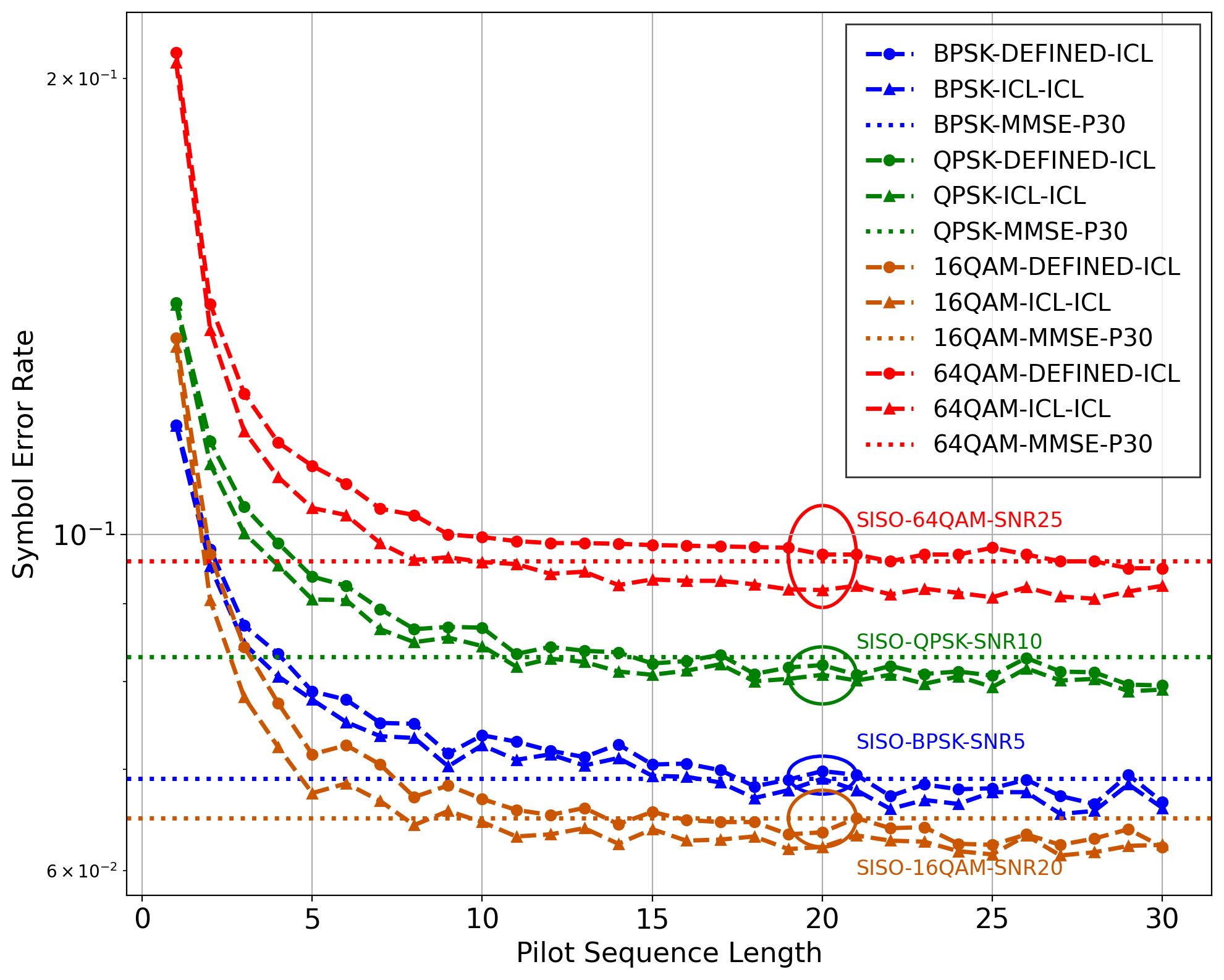}
    \caption{\footnotesize SER as a function of pilot sequence length under ICL-testing conditions, comparing \alg\ and vanilla ICL models in SISO systems across various modulation schemes.}
    \label{fig:ICL_results}
\end{figure}


\subsubsection{Comparison with Different SNRs and Varying Pilot Lengths}
 At high SNR levels, reduced data noise enables more accurate detections from pilot data, enhancing \alg model performance and accentuating the downward SER trend. However, at very high SNRs, the already low initial SER limits further improvement. Our \alg model still performs robustly with minimal pilot data, including the extreme single-pilot cases. As pilot length increases, as shown in \cref{fig:DEFINE_SISO}, all algorithms exhibit improved performance in DF inference.

\subsubsection{Comparison with Different Modulation Schemes}

We conduct experiments using BPSK, QPSK, 16QAM, and 64QAM modulations, which correspond to classification tasks with 2, 4, 16, and 64 classes, respectively. As modulation complexity increases, the detection task becomes more challenging. However, somewhat counterintuitively, we observe larger performance gains with \alg: the SER decreases more sharply with additional feedback data (see Figs.~\ref{fig:DEFINE_SISO}).

Analysis of the Transformer’s output logit vector shows that nearly all of the incorrect detections occur within a small region around the ground-truth label in the constellation. This ``typical error event'' \cite{TV:05} suggests that even incorrect detections carry valuable information, \textit{as the noisy label is often very close to the correct one, which still provides reasonably accurate channel information}. Thus, as modulation complexity increases, the compact constellation set allows noisy feedback to provide more useful information, leading to better SER gains with added feedback data.

These findings demonstrate that \alg effectively captures the constellation set's geometry. It not only learns detections but also recognizes relationships between classes, often assigning higher probabilities to neighboring labels when errors occur. Due to inherent data noise, such as channel fading and additive noise, received signals with nearby latent labels in the constellation may overlap. As a result, the model sees adjacent labels as close neighbors, letting its detections retain valuable information, even when they are not entirely accurate.


\begin{figure}[!htbp]
    \centering
    \includegraphics[width=0.45\textwidth]{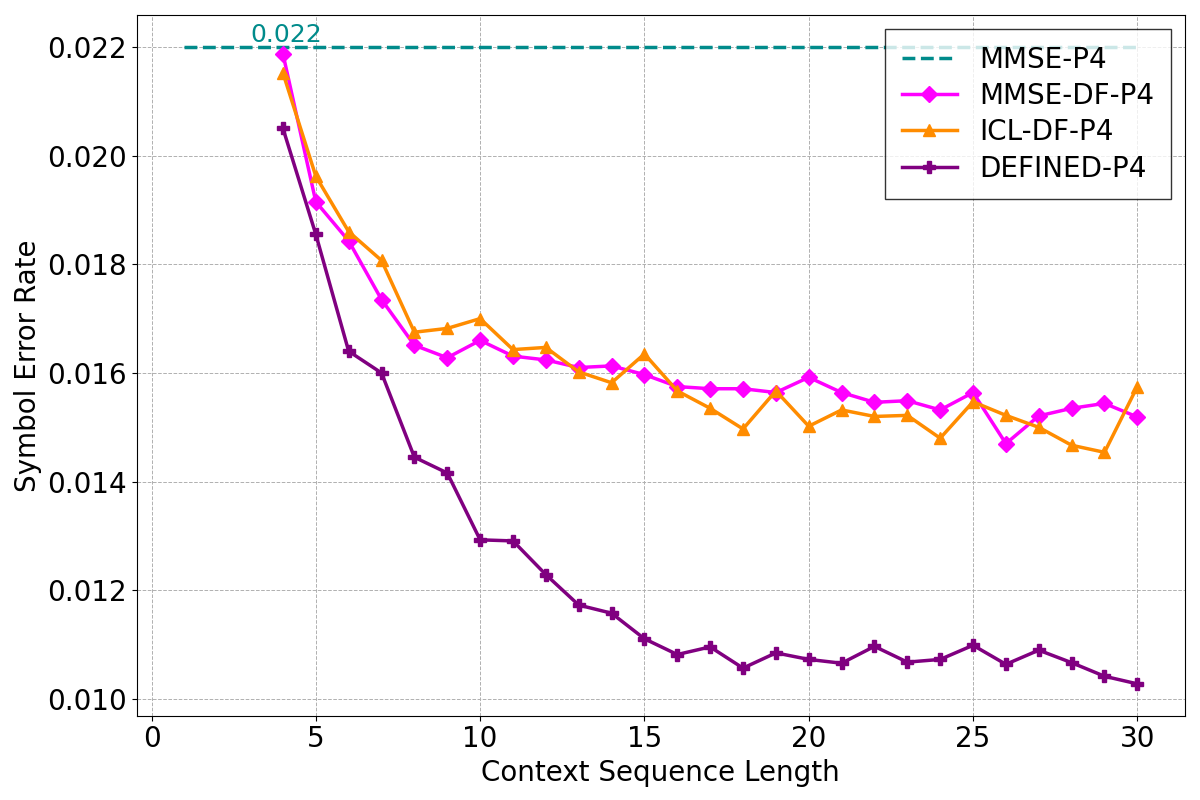}
    \caption{\footnotesize SER as a function of context sequence length for MIMO systems with QPSK at SNR = 15 dB and 4 pilots ($\text{gain}_{\text{DF}} = 50.2\%$).}
    \label{fig:MIMO_QPSK_SNR15_P4}
\end{figure}

\subsubsection{MIMO Results}

The results for MIMO systems with QPSK using 4 pilots are given in \cref{fig:MIMO_QPSK_SNR15_P4}. Compared across different modulation schemes, SNR levels, \alg can still perform well for multi-dimensional detection problems. In fact, \alg exhibits a more pronounced decrease in SER with the inclusion of additional decision feedback data, leveraging the underlying communication system structure more effectively. 


\subsubsection{Robustness to Channel Distribution Mismatch}
\label{sec:exp_mismatch}
As mentioned, handling training and testing channel distribution mismatch \emph{without ML model update} is one of the most important benefits of ICL-based symbol detection, and we validate this feature in this subsection.  

Specifically, we consider a \alg model trained on a SISO 64QAM system with a single pilot and SNRs uniformly sampled from [30, 40] dB, consistent with the previous experiment. 
First, the model is trained on a Rayleigh fading channel without a line-of-sight (LOS) component: $H_{\text{train}} \sim \mathcal{CN}(0,I)$, and tested on a Rician fading channel with an LOS component at SNR = 25 dB:
$H_{\text{test}} \sim \sqrt{\frac{\kappa}{\kappa+1}} e^{j\theta} + \sqrt{\frac{1}{\kappa+1}} \mathcal{CN}(0, I), \label{eqn:Rician}$
where $\kappa = 4$ is the Ricean factor modeling suburban or rural environments \cite{TV:05}, 
and $\theta$ is uniformly sampled to reflect varying LOS conditions. As shown in \cref{fig:OOD_64QAM_Rician_SNR25_P1}, \alg remains robust despite the channel mismatch, benefiting from the additional LOS component, which reduces the SER from $0.118$ to $0.038$. This result also validates the theoretical insight in \cref{remark1}, that the Transformer exhibits inherent robustness to distributional shifts in $\mu_0$ and $\mu_1$, here manifested by the LOS perturbation in the Rician channel.

Second, we evaluate a more challenging scenario where the model is tested in a Rayleigh fading channel with two pilots at SNR = 25 dB (which is below the range of SNRs in the training), introducing increased noise and pilot mismatch. As shown in \cref{fig:OOD_64QAM_SNR25_P2}, \alg achieves a $23.2\%$ SER gain, reducing the SER from $0.138$ to $0.106$. In this extremely noisy condition, MMSE-P30, which uses 30 pilots, achieves an SER of $0.096$. This highlights the adaptability of \alg to varying pilots and noise levels, demonstrating its robustness under adverse conditions.

\begin{figure}[!htbp]
    \centering
    \subfloat[\centering \footnotesize SISO Rician 64QAM, SNR = 25 dB, 1 Pilot, $\text{gain}_{\text{DF}} = 67.8\%$.]{%
        \includegraphics[width=0.45\textwidth]{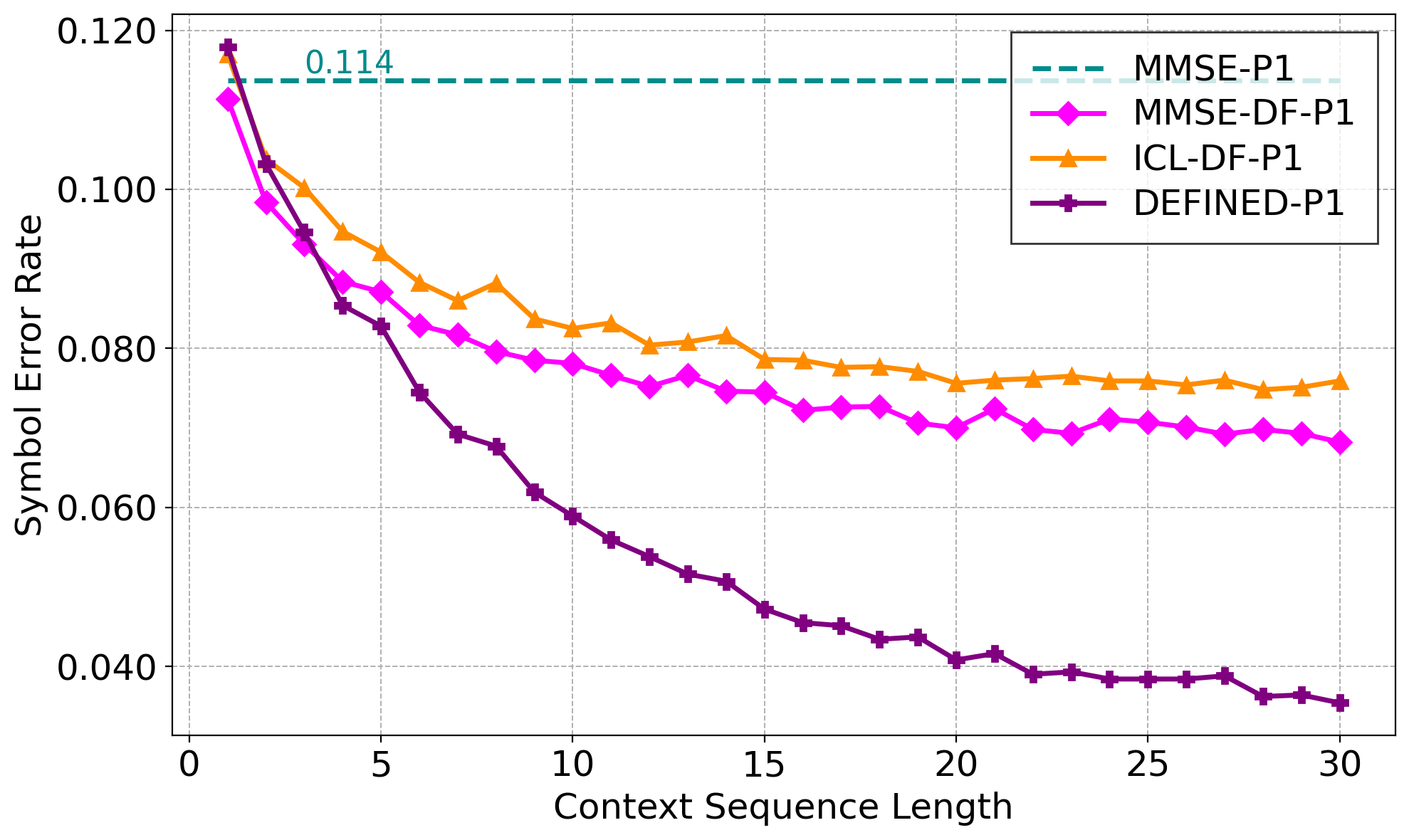}%
        \label{fig:OOD_64QAM_Rician_SNR25_P1}%
    }%
    \qquad
    \subfloat[\centering \footnotesize SISO Rayleigh 64QAM, SNR = 25 dB, 2 Pilots, $\text{gain}_{\text{DF}} = 23.2\%$.]{%
        \includegraphics[width=0.45\textwidth]{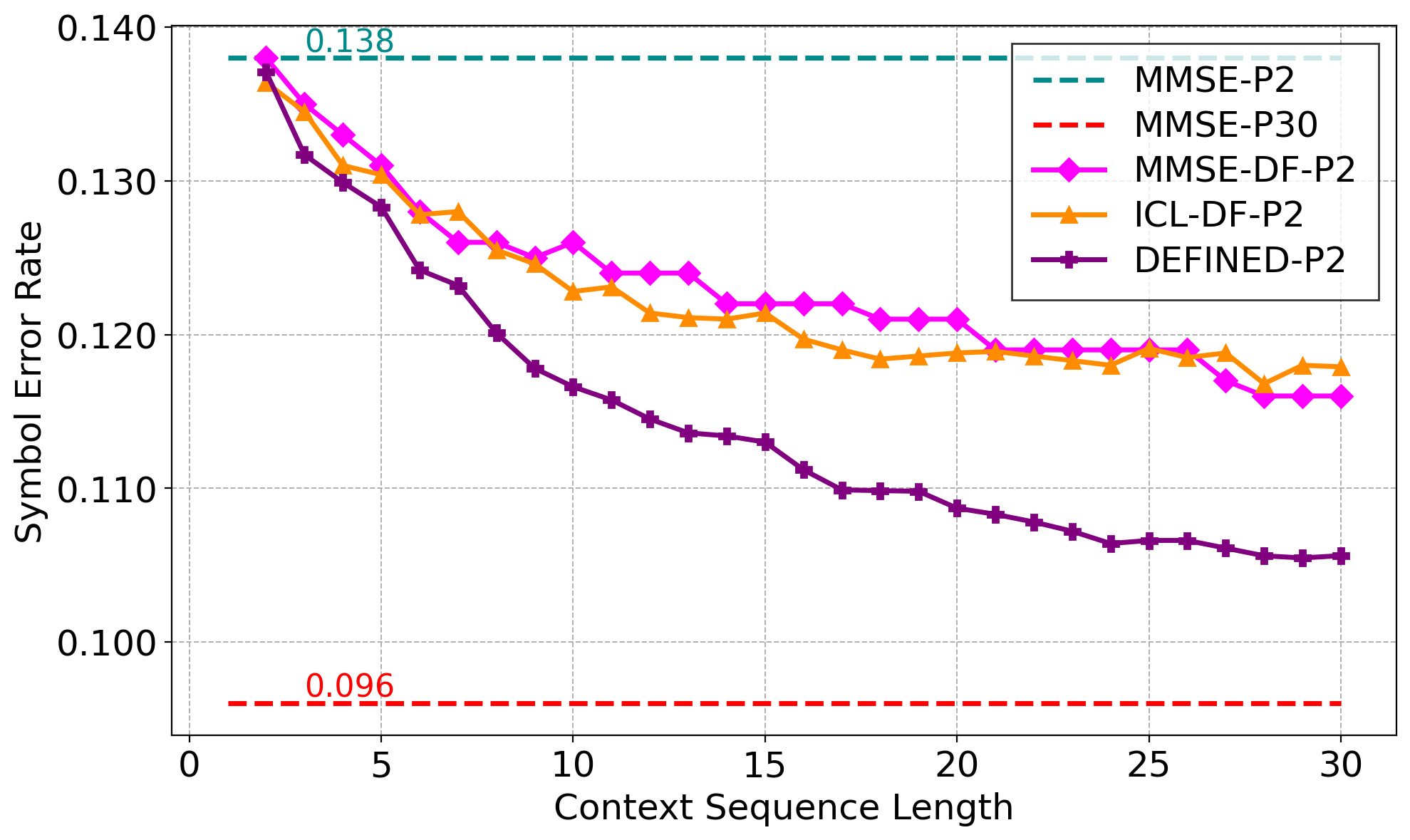}%
        \label{fig:OOD_64QAM_SNR25_P2}%
    }%
    \caption{Performance under channel distribution mismatch.}
    \label{fig:OOD_exp}
\end{figure}

\subsection{Comparison with Non-Coherent Detection}

Here, we compare \alg with a non-coherent detection algorithm, which also does not explicitly perform channel estimation.

\subsubsection{Maximum Likelihood Sequence Detection Algorithm}
\label{section:MLalg}
We consider a non-coherent detection algorithm \cite{papailiopoulos2013maximum} for the block fading channel in the SISO system, which does not require explicit channel estimation. Maximum Likelihood sequence detection (MLSD) is optimal in terms of minimizing the \emph{sequence} error rate on the entire received signal sequence. Given the received signal sequence $Y = [y_1, y_2, \cdots, y_T]$, the MLSD outputs the sequence $S = [s_1, s_2, \cdots, s_T]$ that maximizes the conditional probability density function of $Y$ given $S$. The optimal decision is given by: $S^{\text{MLSD}} = \arg\max_{S\in \mathcal{X}^{T}} f(Y|S)$, 
where $f(Y|S)$ represents the probability density function (PDF) of the channel output conditioned on the transmitted symbol sequence. We consider the Rayleigh block fading channel $H \sim \mathcal{CN}(0,1)$ and channel noise $Z= [z_1,z_2\cdots,z_T]$ with zero mean and covariance matrix $\sigma^2 I$. Thus, $S^{\text{MLSD}}$ can be written as \cite{papailiopoulos2013maximum}:
\begin{equation}
S^{\text{MLSD}} 
= \arg \max_{S \in \mathcal{X}^T} \bigg\{ 
\frac{ \left| S Y^{\dagger} \right|^2 }{\sigma^2 \|S\|^2 + \sigma^4 } 
 - \ln \left( \|S\|^2 + \sigma^2 \right) \bigg\}.
\end{equation}

Two forms of ambiguity exist in this non-coherent detection problem. The first is \emph{phase ambiguity}, which occurs for any constellation invariant under a specific phase rotation. The second is \emph{divisor ambiguity}, which arises when multiple points lie on the same one-dimensional subspace \cite{ryan2006blind}. For SISO scenarios with BPSK and QPSK, since all constellation points lie on the unit circle, the second ambiguity does not exist. The first ambiguity can be resolved by using the phase of the last symbol from the previous block, as described in \cite{chen2003joint}. Practically, this is achieved by overlapping successive blocks by one symbol and using the last symbol of the previous block as the reference for the current block. 

In our experiment, phase ambiguity is resolved by artificially injecting the phase of the first symbol in the transmitted sequence at the receiver and searching for the optimal $S^{\text{MLSD}}$ within the subspace of $\mathcal{X}^{T}$. The average SER is then calculated for the remaining $(T-1)$ symbols, as the phase of the first symbol is known for a coherent interval of length $T$. We note that this is a ``genie'' case that represents the best case scenario. Moreover, knowing the phase of the first symbol in SISO BPSK and QPSK scenarios is \emph{almost} equivalent to having one pilot, and is analogous to the setup used in previous baseline algorithms that include one pilot and a total decision feedback enhanced context sequence of length $(T-1)$. This ensures a relatively fair comparison between the methods. Consequently, we enumerate $T$ from 2 to 31, compute the mean {SER} over tasks for detecting a sequence of $T$ received signals, and plot it at the index corresponding to a context sequence length of $(T-1)$ in the following results.


\begin{figure}[!htbp]
    \centering
    \includegraphics[width=0.45\textwidth]{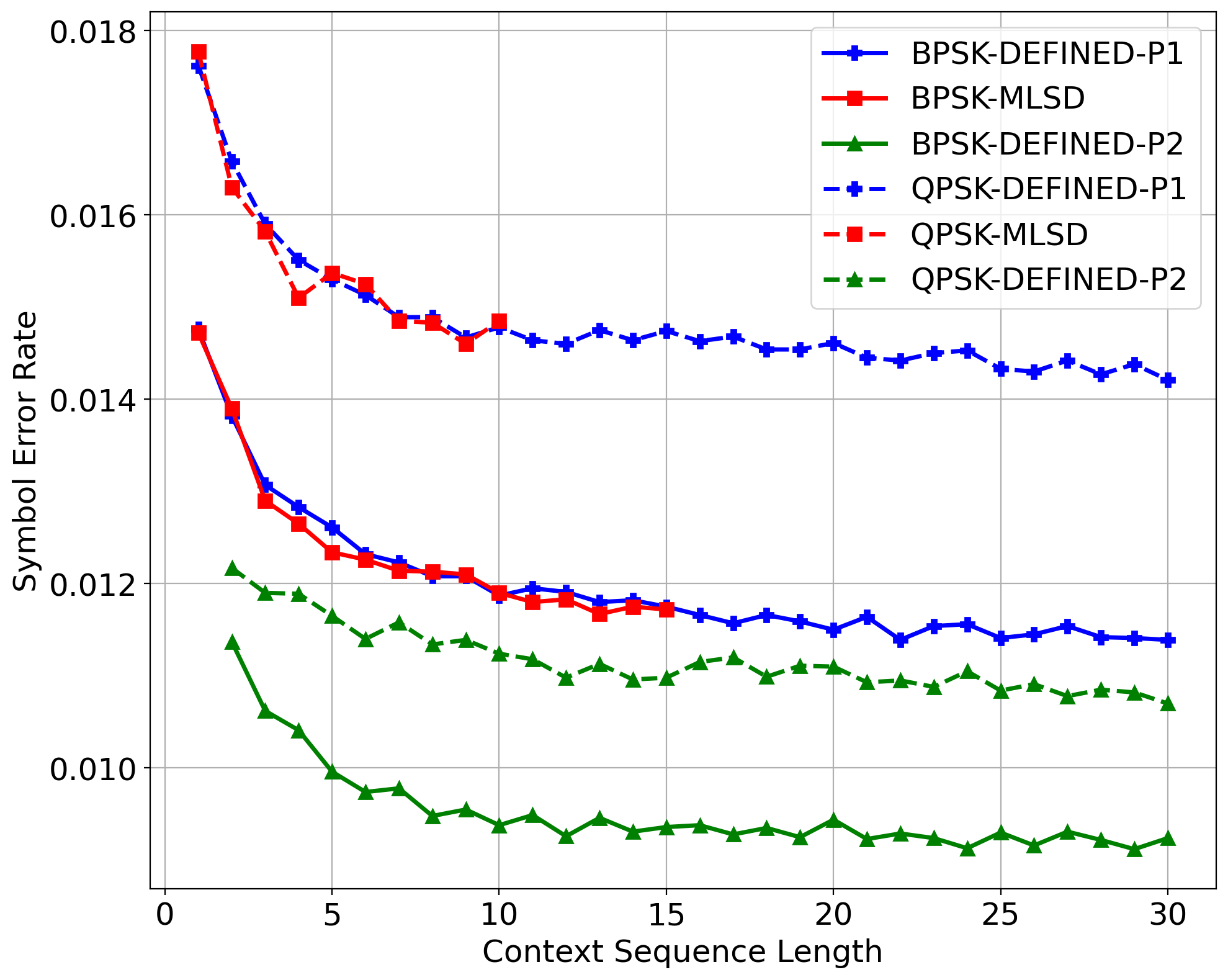}
    \caption{Comparison between \alg\ with one and two pilot pairs and the MLSD algorithm in SISO systems for BPSK at SNR = 15 dB and QPSK at SNR = 20 dB.}

    \label{fig:SISO_MLSD}
\end{figure}

As shown in \cref{fig:SISO_MLSD}, we compare the performance of the non-coherent MLSD algorithm with \alg in the SISO system using BPSK and QPSK modulations. {The implementation of efficient MLSD algorithms has been an active area of research for decades \cite{papailiopoulos2013maximum}. However, since the problem is fundamentally NP-hard, computationally efficient algorithms often suffer from performance degradation (albeit small in some settings). We thus implement MLSD using an exhaustive search for varying sequence lengths. As the computational complexity grows \emph{exponentially} with the sequence length, we only report the feasible results within some reasonable timeframe.}

The DEFINED-P1 curves demonstrate that our model’s performance with one pilot under DF-testing matches that of the MLSD detector, which is the optimal {sequence} estimator with the {additional} first symbol phase information (almost equivalent to one pilot). Furthermore, with two pilots, \alg significantly outperforms the MLSD detector, as shown by the {DEFINED-P2} line. The adaptability of \alg allows it to smoothly transition between non-coherent and coherent detection modes (via varying the amount of pilots) while consistently benefiting from decision feedback data.

Our experiment results demonstrate that \alg with one pilot empirically aligns with the MLSD estimator at short sequence lengths. This empirical finding aligns with previous theoretical results on ICL \cite{shen2024training, panwarcontext, raventos2024pretraining}. Practically, although \alg achieves similar or slightly better performance than MLSD during a single forward pass, we can avoid the exponentially increasing computational complexity associated with a true MLSD. Moreover, as the number of pilot pairs increases, \alg smoothly transitions to coherent detection and significantly outperforms the MLSD detector. Additionally, the true data distribution is often infeasible to obtain in real-world scenarios, making MLSD detection challenging to implement. In contrast, data-driven methods like Transformers show great potential in wireless communication tasks.

\section{Conclusions}
\label{sec:conc}

Inspired by decision feedback in wireless communication, we proposed \alg to enhance symbol detection by incorporating decision feedback into in-context prompts. Our approach achieved significant performance gains with limited pilot data while maintaining high accuracy when pilot information is sufficient, making it highly applicable to practical scenarios. Extensive experiments across various modulation schemes confirmed the robustness and adaptability of our model. Beyond empirical performance, we provided theoretical insights into the generalization behavior of pre-trained Transformers under SNR mismatch and derived an error lower bound for the \alg model. These findings highlight the potential of Transformer architectures as powerful in-context learners for wireless systems and motivate further theoretical exploration of their limits and capabilities.

\appendix

\begin{proof}[Proof of Theorem~\ref{thm: lower bound}]
We denote 
$\mu=\mu_1-\mu_0$,
$u=2(\mu_1+\mu_0)$,
$q=y_\query$. Define $p=\frac{2}{k} \sum_{i=1}^k y_{i}x_{i}$. Since with probability $\probp{x_{i}=1}=1/2$, $y_{i}=\mu_{1}+v_i$, with probability $\probp{x_{i}=-1}=1/2$, $y_{i}=\mu_{0}+v_i$, where $v_i\sim\normal(0,\Lambda)$, we have $p=2k_1\mu_{1}/k-2k_0\mu_{0}/k+g$, where $g=\frac{2}{k}\sum_{i=1}^kv_i$, $g\sim\normal(0,4\Lambda/k)$, 
$k_1$ is a binomial random variable.
Defining $h= k_1/k-1/2$,  we have 
$p=\mu+hu +g$.
The output of the trained Transformer is
\begin{align}
    \widehat{x}_\out=S\pth{\pth{\frac{2}{k} \sum_{i=1}^k \testy_i \testx_i^\top}\Lambda^{-1}\testy_\query}=S(p^\top \Lambda^{-1}q).
\end{align}
The probability of $x_\query=1$ given $y_\query$ is 
    $$\probp{x_\query=1 | y_\query}=S((\mu_1-\mu_0)^\top \Lambda^{-1} x_\query)=S(\mu^\top\Lambda^{-1}q).$$

Defining $a=\mu^\top  \Lambda^{-1} q$, $b=(hu +g)^\top\Lambda^{-1}q$,  we have
    \begin{align*}
        &p^\top \Lambda^{-1} q=(\mu+hu +g)^\top\Lambda^{-1}q\\
        =&(hu +g)^\top\Lambda^{-1}q+ \mu^\top\Lambda^{-1}q=a+b,
    \end{align*}
and $S(p^\top \Lambda^{-1}q)=S(a+b)=S(a)+S'(a)b+S''(\xi(a,b))b^2/2$, 
where $\xi$ are real numbers between $a$ and $a+b$.
Thus, we have
\begin{align*}
    &\E[\pth{S(a+b)-S(a)}^2]=\E[\pth{S'(a)b+S''(\xi(a,b))b^2/2}^2]\\
    =& \E[(S'(a))^2b^2]+\E[S'(a)S''(\xi(a,b))b^3+(S''(\xi(a,b)))^2b^4/4]
\end{align*}
Note that $\E[(S'(a))^2b^2]=(S'(a))^2\E[b^2]$. Let us first consider: 
\begin{align*}
    \E[b^2]=&\E[h^2 u^\top \Lambda^{-1} q u^\top \Lambda^{-1} q + g^\top \Lambda^{-1} qg^\top \Lambda^{-1}q]\\
    \overset{(a)}{=}&(u^\top \Lambda^{-1} q)^2 /(4k) + 4q^\top \Lambda^{-1} q/k
\end{align*}
where $(a)$ is due to Lemma D.1 in \cite{shen2024training} that $\E[h^2]=1/(4k)$, $g^\top \Lambda^{-1} qg^\top \Lambda^{-1}q=q^\top\Lambda^{-1} gg^\top \Lambda^{-1}q=(gg^\top \Lambda^{-1} q)^\top \Lambda^{-1}q$ and $\E[gg^\top]=4\Lambda/k$.

For term $\E[S'(a)S''(\xi(a,b))b^3+(S''(\xi(a,b)))^2b^4/4]$, we have 
\begin{align*}
    \E[|S'(a)S''(\xi(a,b))b^3+(S''(\xi(a,b)))^2b^4/4|]\leq& \E[|b^3|]+\E[b^4].
\end{align*}
Terms in $\E[|b^3|]$ contain three $h$, or three $g$, or one $h$ and two $g$, or two $h$ and one $g$. According to Lemma D.1, we have $\E[|h|]=O(1/\sqrt{k})$, $\E[h^2]=1/(4k)$, $\E[|h^3|]= O(k^{-3/2})$. Moreover,  $g=2k^{-1/2}\Lambda^{1/2}\Bar{g}$ and $\Bar{g}\sim \normal(0,I)$. Converting one $g$ to $\Bar{g}$, we have a coefficient of $k^{-1/2}$. Thus, terms in $\E[|b^3|]$ are $O(k^{-3/2})$. Similarly, we have $\E[b^4]=O(k^{-2})$. Thus, we have
\begin{align*}
    &\E[|S'(a)S''(\xi(a,b))b^3+(S''(\xi(a,b)))^2b^4/4|]=o(1/k).
\end{align*}

Putting all together, we have: 
\begin{align*}
    &\E[|\widehat{x}_\out-\probp{x_\query=1 | y_\query}|^2]=\E[\pth{S(a+b)-S(a)}^2]\\
    =&\frac{1}{k}(S'(a))^2[(u^\top \Lambda^{-1} q)^2 /4 + 4q^\top \Lambda^{-1} q]+o(1/k).
\end{align*}
\end{proof}


\begin{proof}[Proof of Theorem~\ref{thm:mismatch}]
    Suppose a random variable $x$ conditioned on $y$ has distribution $\probp{x=1|y}=p$ and $\probp{x=-1|y}=1-p$. Suppose the distribution of a predictor $\tilde x$ conditioned on $y$ is $\probp{\tilde x=1|y}=q$ and $\probp{\tilde x=-1|y}=1-q$. Given $y$, the accuracy between $x$ and $\tilde x$ is $Acc(x, \tilde x)=pq+(1-p)(1-q)$, where $p,q\in[0,1]$.
    We can easily find that the best predictor $\tilde x$ that maximize the accuracy is
    \begin{align*}
        \tilde x =\left\{ \begin{array}{ll}
            1, &if p>1/2\\
            -1, &if   p\leq 1/2       \end{array} \right.
    \end{align*}
    Note that in our setting, $\probp{x_\query=1 | y_\query}=S(\sigma^{-2}(\mu_1-\mu_0)^\top y_\query)$. Thus,
    the best predictor $x_p$ that maximize the accuracy between $ x_p$ and $x_\query$ should be
    \begin{align*}
        x_p=\left\{ \begin{array}{ll}
            1, &if (\mu_1-\mu_0)^\top y_\query>0\\
            -1, &if  (\mu_1-\mu_0)^\top y_\query\leq 0       \end{array} \right.
    \end{align*}
    Moreover, according to Theorem \ref{thm: lower bound}, when $k\to\infty$, 
    $\widehat x_\out\to S(\xi^{-2}(\mu_1-\mu_0)^\top y_\query).$ 
    Thus, if we define $\widehat x_p$ as 
    \begin{align*}
        \widehat x_p=\left\{ \begin{array}{ll}
            1, &if x_\out>1/2\\
            -1, &if  x_\out\leq 1/2       \end{array} \right.
    \end{align*}
    Then, $\widehat x_p$ is equivalent to $x_p$ when $k\to\infty$.
\end{proof}

\bibliographystyle{IEEEtran}
\bibliography{ref, Shen, wireless}

\end{document}